%% file: directed.tex
\documentclass[sigconf]{acmart}

\usepackage{macros}

\usepackage{enumitem}
\usepackage{microtype}
\usepackage{algorithmicx}
\usepackage[linesnumbered,ruled,vlined]{algorithm2e}

\usepackage{amsthm}

\newtheorem*{theorem*}{Theorem}
\newtheorem*{definition*}{Definition}

\setcopyright{rightsretained}

\usepackage{paralist}

\renewenvironment{enumerate}[1]{\begin{compactenum}#1}{\end{compactenum}}

\begin{document}

\title{Measuring Directed Triadic Closure with Closure Coefficients}

\author{Hao Yin}
\affiliation{%
  \institution{Stanford University}
}
\email{yinh@stanford.edu}

\author{Austin R.\ Benson}
\affiliation{%
  \institution{Cornell University}
}
\email{arb@cs.cornell.edu}

\author{Johan Ugander}
\affiliation{%
  \institution{Stanford University}
}
\email{jugander@stanford.edu}


\renewcommand{\shortauthors}{Hao Yin et al.}

\begin{abstract}
\input{000abstract.tex}
\end{abstract}


\maketitle


\section{Introduction}
\label{sec:introduction}
\input{010introduction}

\section{Background and preliminaries}
\label{sec:preliminaries}
\input{020preliminaries}

\section{Directed Closure Coefficients}
\label{sec:def}
\input{030def}

\section{Theoretical analysis}
\label{sec:theory}
\input{040theory}
\input{041symmetry}
\input{042cm_theory}

\section{Case study in node-type prediction}
\label{sec:application}
\input{050application}



\section{Conclusion}
\label{sec:discuss}
\input{070discuss}

\begin{acks}
This research has been supported in part by an ARO Young Investigator Award,
NSF Award DMS-1830274,
and ARO Award W911NF-19-1-0057.
\end{acks}

\bibliographystyle{ACM-Reference-Format}
\bibliography{refs,more-refs}


\end{document}

%% file: 000abstract.tex

Recent work studying triadic closure in undirected graphs has drawn attention 
to the distinction between measures that focus on the ``center'' 
node of a wedge (\ie, length-2 path) vs.\ measures that focus on the ``initiator,'' a distinction 
with considerable consequences. 
Existing measures in directed graphs, meanwhile, have all been 
center-focused. 
In this work, we propose a family of eight \emph{directed closure coefficients} that  
measure the frequency of triadic closure in directed graphs from the perspective 
of the node initiating closure. 
The eight coefficients correspond to different labelled wedges, where the 
initiator and center nodes are labelled, and
we observe dramatic empirical variation in these coefficients on real-world 
networks, even in cases when the induced directed triangles are isomorphic. 
To understand this phenomenon, we examine the theoretical behavior of our 
closure coefficients under a directed configuration model. Our analysis illustrates 
an underlying connection between the closure coefficients and moments of the joint 
in- and out-degree distributions of the network, offering an explanation of the observed asymmetries. 
We also use our directed closure coefficients as predictors in two machine 
learning tasks. We find interpretable models with AUC scores above 0.92 in 
class-balanced binary prediction, substantially outperforming models that use
traditional center-focused measures.

\hide{
Recent work studying clustering in undirected graphs has drawn attention 
to the distinction between measures of clustering that focus on the ``center'' 
node of a triangle vs.~measures that focus on the ``initiator,'' and this distinction 
has been found to have considerable consequences. 
Meanwhile, the limited existing approaches to clustering in directed graphs 
have all been center-focused. 
In this work, we propose a family of \emph{directed closure coefficients} that 
measure the frequency of triadic closure in directed graphs from the perspective 
of the node initiating closure. We again find considerable consequences stemming 
from this slight change of definition, including new consequences that are unique 
to the directed setting.

We observe dramatic variation in the frequency of triadic closure measured by
different directed closure coefficients in real-world graphs, even in cases when 
the induced directed triangles are isomorphic. 
To better understand this phenomenon, we provide an analysis showing that these 
asymmetries can be arbitrarily extreme. 
We further examine the theoretical behavior of our directed closure coefficients
under the directed configuration model and observe an underlying connection
between the closure coefficients and moments of the joint in- and out-degree
distributions that can explain the observed asymmetries. 
Finally, we use our directed closure coefficients as predictors in two network 
machine learning tasks and find that interpretable models achieve mean AUC 
scores above 92\% in class-balanced binary prediction, substantially outperforming 
models that use other predictors related to triadic closure (77\%).
}

\medskip
{\noindent \bf Keywords:} 
directed networks,
triadic closure,
closure coefficients,
configuration model

%% file: 010introduction.tex


A fundamental property of networks across domains is the increased probability 
of edges existing between nodes that share a common neighbor, a phenomenon known 
as triadic closure~\cite{simmel1908soziologie,rapoport1953spread,watts1998collective}. 
This concept underpins various ideas in the study of networks---especially 
in undirected network models with symmetric relationships---including the
development of generative models~\cite{leskovec2005graphs, jackson2007meeting, seshadhri2012community, robles2016sampling}, 
community detection methods~\cite{fortunato2010community, gleich2012vertex}, and feature
extraction for network-based machine learning tasks~\cite{henderson2012rolx,lafond2014anomaly}.


A standard measure for the frequency of triadic closure on undirected networks is the
\emph{clustering coefficient}~\cite{watts1998collective,barrat2000properties,newman2001random}.
At the node level, the \emph{local clustering coefficient} of a node $u$ is 
defined as the fraction of wedges (\ie, length-2 paths) with center $u$ that are \emph{closed},
meaning that there is an edge connecting the two ends of the wedge, inducing a triangle.
At the network level, the \emph{average clustering coefficient} is the
mean of the local clustering coefficients~\cite{watts1998collective},
and the \emph{global clustering coefficient}, also known as \emph{transitivity}~\cite{barrat2000properties,newman2001random},
is the fraction of wedges in the entire network that are closed.

Recent research has pointed out a fundamental gap between how
triadic closure is measured by the clustering coefficient and how it is usually explained~\cite{yin2019local}.
Local triangle formation is usually explained by some transitive property of 
the relationships that edges represent; for social networks, this is embodied
in the idea that ``a friend of my friend is my friend''.
In these explanations, however, triadic closure is driven not by the center of a length-2 path
but rather by an end node (which we refer to as the \emph{head}),
who initiates a new connection. In contrast, the local clustering coefficient
that measures triadic closure from the center of a wedge implicitly accredits the closure
to the center node. The recently proposed \emph{local closure coefficient} closes this
definitional gap for undirected graphs 
by measuring closure with respect to the fraction of length-2 paths
starting from a specified head node that are closed~\cite{yin2019local}.

These closure coefficients were only defined on undirected networks, but
the interactions in many real-world networks are more accurately modeled
with an associated orientation or direction.
Examples of such networks include food webs,
where the direction of edges represents carbon or energy flow from one ecological
compartment to another; hyperlink graphs, where edges represent which web pages
link to which others; and certain online social networks such as Twitter, where ``following''
relationships are often not reciprocated~\cite{cheng2011predicting}.
The direction of edges may reveal underlying hierarchical structure in a 
network~\cite{homans1950human,davis1967structure,ball2013friendship}, 
and we should expect the direction to play a role in local triadic closure.

Extensions of clustering coefficients have been proposed in directed 
networks~\cite{fagiolo2007clustering,seshadhri2016directed}, which are center-based
at the node level.
However, formulating directed triadic measures from the center of a wedge 
is even less natural in the directed case, while 
measuring from the head is a more common description of directed closure relationships.
For example, in citation networks,
paper $A$ may cite $B$, which cites $C$ and leads $A$ to also cite $C$.
In this scenario, the initiator of this triadic closure is really paper $A$.
Similarly, in directed social networks, outgoing edges may represent 
differential status~\cite{leskovec2010signed,ball2013friendship},
where if person $A$ thinks highly of $B$ and $B$ thinks highly of $C$, 
then $A$ is likely to think highly of person $C$ and 
consequently initiate an outbound link.

In the above examples, measuring triadic closure from $A$ would be the analog of the closure coefficient
for directed networks, which is what we develop in this paper.
More specifically, we propose a family of directed closure coefficients, which are natural generalizations 
of the closure coefficient for undirected networks. Like the undirected version of closure 
coefficients, these measures are based on the head node of a length-2 path, 
in agreement with common mechanistic interpretations 
of directed triadic closure and fundamentally
different from the center-based clustering coefficient.
Specifically, the {\it directed clustering coefficients} proposed by Fagiolo ~\cite{fagiolo2007clustering} 
are not to be confused with the {\it directed closure coefficients} introduced in this work.

\begin{figure}[t]
\centering
\includegraphics[width=\columnwidth]{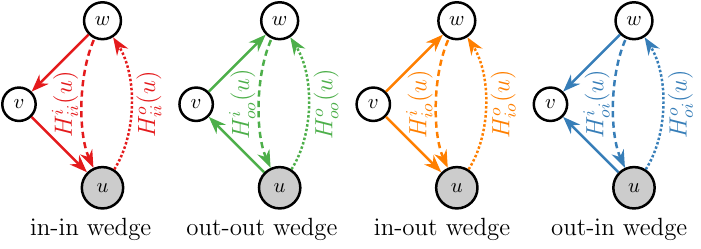} 
\caption{%
Illustration of four wedge types and eight local directed closure coefficients at node $u$.
The type of wedge is denoted by two letters, each representing an edge direction (incoming or outgoing).
The first letter represents the direction of the edge between the head node $u$ and the center node $v$ with respect to $u$.
The second letter represents the direction of the edge between $v$ and the tail node $w$ with respect to $v$.
A wedge is \emph{$i$-closed} if there is an incoming edge to the head node from the tail,
and \emph{$o$-closed} if there is an outgoing edge from the head node to the tail.
There are eight local directed closure coefficients at node $u$, denoted as $\lcc{xy}{z}{u}$
with $x, y, z \in \{i, o\}$. Each local directed closure coefficient measures the frequency of
triadic closure of a certain wedge type (denoted by subscript $xy$) and
closing direction (denoted by superscript $z$).
}
\label{Fig:DefCC}
\end{figure}

Our measurements are based on the notion of a \emph{directed wedge} 
as an ordered pair of directed edges that share a common node,
and the ``non-center'' end nodes of this wedge on the first and second edge
are called the \emph{head} and \emph{tail} nodes, respectively 
(in \Cref{Fig:DefCC}, solid lines mark the wedge, where node $u$ is the head and node $w$ is the tail).
Since each edge may be in either direction, there are four directed wedge types.
When considering triadic closure for each wedge type,
the closing edge between the head and tail nodes may also take either direction. 
Therefore, at each node, there are eight local directed closure coefficients, each representing the 
frequency of directed triadic closure with a certain wedge type and closure direction (\Cref{Fig:DefCC}).
Analogous to the undirected case, we also define the average and global directed closure
coefficients to measure the overall frequency of triadic closure in the entire network.
These statistics provide a natural and intuitive way to study the frequency of directed triadic closure in detail,
including how directions of the incident and second edge
influence a node's tendency to initiate or receive directed triadic closure.

Our empirical evaluation of the directed closure coefficients on real-world networks
in \cref{sec:empirical}
reveals several interesting patterns. At the node level, we find clear evidence of a 2-block correlation
structure amongst the eight local directed closure coefficients, where coefficients within
one block are positively (but not perfectly) correlated while coefficients from 
distinct blocks are nearly uncorrelated.
The block separation coincides with the direction of the closing edge in the closure coefficients.
We also provide theoretical justification for this observation, gleaned from studying the expected
behavior of the closure coefficients for directed configuration model random graphs.
Specifically, we will show that the expected value (under this model) 
of each local directed closure coefficient increases 
with the node degree in the closing edge direction, and thus coefficients with the same closure
direction and directed degree are correlated.

From empirical network measurements, we also find surprising asymmetry amongst average closure
coefficients. Consider the in-out wedges in \Cref{Fig:DefCC}, where the coefficients
$\lcc{io}{i}{u}$ and $\lcc{io}{o}{u}$ correspond to the same directed induced subgraph.
For such symmetric wedges, the likelihood for outbound closure can be
substantially higher than for inbound, even though the two induced subgraphs are structurally identical.
On the other hand, we show
in \cref{subsec:symmetry}
that networks from the same domain exhibit the same asymmetries.

With extremal analysis, we show
in \cref{subsec:symmetry}
that there is in fact no positive lower or upper bound on the ratio between 
types of directed average closure coefficients. 
Additional probabilistic analysis under the configuration model 
shows that the expected values of the directed
closure coefficients depend on various second-order moments of the joint
in- and out- degree distribution of the network. 
This result partly explains the significant difference in values between
a pair of seemingly related average closure coefficients: their expected behaviors
correspond to different second order moments of the degree distribution.

Beyond our intrinsic study on the structure of directed closure coefficients, 
we illustrate
in \cref{sec:application}
how these coefficients can be powerful features for
network-based machine learning. In a lawyer advisory
network where every node (lawyer) 
is labeled with a status level (partner or associate) 
and directed edges correspond to who talks to whom for profession advice~\cite{lazega2001collegial}, we
show that local directed closure coefficients are much better predictors of status compared 
to other structural features such as degree or Fagiolo's directed clustering coefficients.
Analysis of the regularization path of the predictive model yields the insight that it is not 
\emph{how many one advises} but rather \emph{who one advises}
that is predictive of partner status.
We conduct a similar network classification task in an entirely different domain using a food web
from an ecological study. Using the same tools, we find that directed closure coefficients are good
predictors of whether or not a species is a fish. This highlights how our proposed measurements
are potentially useful across many domains.

In summary, we propose the directed closure coefficients, a family of eight new metrics for directed triadic
closure on directed networks. We provide extensive theoretical analysis which help explain some
counter-intuitive empirical observations on real-world networks. Through two case studies, we demonstrate that our 
proposed measurements are good predictors in network-based machine learning tasks.

%% file: 020preliminaries.tex

An undirected network (graph) $G = (V, E)$ is a node set $V$ and an edge set $E$,
where an edge $e \in E$ connects two nodes $u$ and $v$.
We use $\degree{u}$ to denote the degree of node $u \in V$, \ie, the number of edges adjacent to $u$.
A \emph{wedge} is an ordered pair of edges that share exactly one node;
the shared node is the \emph{center} of the wedge.
A wedge is \emph{closed} if there is an edge connecting the two non-center nodes (\ie,
the nodes in the wedge induce a triangle in the graph).

Although the notion of triadic closure in general has a long history~\cite{rapoport1953spread,wasserman1994social},
perhaps the most common metric for measuring triadic closure in undirected
networks is the average clustering coefficient~\cite{watts1998collective}.
This metric is the mean of the set of \emph{local clustering coefficients}
of the nodes, where the local clustering coefficient of a node $u$, $\ulccc{}{u}$,
is the fraction of wedges centered at node $u$ that are closed:
\[ 
\ulccc{}{u} = \frac{2  \triad{u} } { \degree{u} \cdot ( \degree{u} - 1)},
\]
where $T(u)$ denotes the number of triangles in which node $u$ participates.
The denominator $\degree{u} \cdot ( \degree{u} - 1)$ is the number of wedges
centered at $u$, and the coefficient 2 corresponds to the two wedges (two
ordered pairs of neighbors) centered at $u$ that the triangle closes. If there
is no wedge centered at $u$ (\ie, $\degree{u} \leq 1$), the local
clustering coefficient is undefined.

Again, to measure the overall triadic closure of the entire network, 
the {\em average clustering coefficient} is defined as the 
mean of the local clustering coefficients of all nodes:
\[
\uaccc{} = \frac{1}{\lvert V \rvert}\sum_{u \in V} \ulccc{}{u}.
\]
%
When undefined, the local clustering coefficient is treated as zero in this average~\cite{newman2003structure},
although there are other ways to handle these cases~\cite{kaiser2008mean}.
An alternative network-level version of the clustering coefficient is the \emph{global clustering coefficient},
which is the fraction  of closed wedges in the entire
network~\cite{barrat2000properties,newman2001random},
\[ 
\ugcc{} = \frac{2 \sum_{u \in V}\triad{u} } { \sum_{u \in V} \degree{u} \cdot ( \degree{u} - 1)}.
\]
This measure is also sometimes called \emph{transitivity}~\cite{boccaletti2006complex}.

Recent research has exposed fundamental differences in how triadic closure is interpreted and measured~\cite{yin2019local}.
For example, social network triadic closure is often explained by the old adage that 
``a friend of a friend is my friend,'' which accredits the creation of the third edge to the end-node (also called the 
{\em head}) of the wedge. This interpretation, however, is fundamentally at odds with how triadic closure is measured
by the clustering coefficient, which is from the perspective of the center node. To close this gap, Yin \etal 
proposed the {\em local closure coefficient} that measures triadic closure from the head node of wedges~\cite{yin2019local}.
Formally, they define this as
\[
\ulcc{}{u} = \frac{2  \triad{u} }{ \sum_{v \in N(u)} [\degree{v} - 1]},
\]
where $N(u)$ is the set of neighbors of $u$.
In this case, the denominator is the number of length-2 paths emanating 
from node $u$.
Thus, in social networks, the closure coefficient of a node $u$ 
can be interpreted as the fraction of friends of friends of $u$ that are themselves friends with $u$.
The closure coefficient has since been investigated under scale-free random graph 
models~\cite{stegehuis2019closure}.

\xhdr{Extensions to directed networks}
The focus of this paper is on measuring triadic closure in directed networks.
The only definitional difference from undirected networks is that
the edges are equipped with an orientation, and $(u, v) \in E$ denotes
a directed edge pointing from $u$ to $v$.%
\footnote{One nuance in directed networks is that an edge might be reciprocal:
$(u,v) \in E$ and $(v, u) \in E$. A pair of reciprocal edges are 
sometimes treated as a single bidirected edge~\cite{garlaschelli2004patterns,seshadhri2016directed}
and sometimes treated as two distinct edges~\cite{sarajlic2016graphlet}.
For readability proposes, in this paper we 
treat a pair of reciprocal edges as two separate edges.
Extensions for special considerations of reciprocal edges
are straightforward and similar theoretical and empirical results can be found.
}
We assume that $G$ does not contain multi-edges or self-loops
and denote the number of nodes by $n = \lvert V \rvert$ and the
number of edges by $m = \lvert E \rvert$.

When an end-node $u$ of an edge is specified, we denote the direction of 
an edge as $i$ (for incoming to $u$) or $o$ (for outgoing from $u$). 
For any node $u \in V$, we use $\indeg{u}$ and $\outdeg{u}$ to denote 
its \emph{in-degree} and \emph{out-degree},
\ie, the number of edges incoming to and outgoing from node $u$, respectively.
For a sequence of joint in- and out-degrees $[(\indeg{u}, \outdeg{u})]_{u \in V}$, 
we use $M_{xy}$, with $x, y \in \{i, o\}$
being the direction indicator, to denote the different second-order moments of the degree
sequence, \ie, 
\[ 
M_{xy} = \frac{1}{n} \sum_{u \in V} \dirdeg{x}{u} \dirdeg{y}{u}.
\]
There are three second-order moments: $M_{ii}$, $M_{oo}$, and $M_{io} = M_{oi}$.

\begin{figure}[t]
\centering
\includegraphics[width=3.2in]{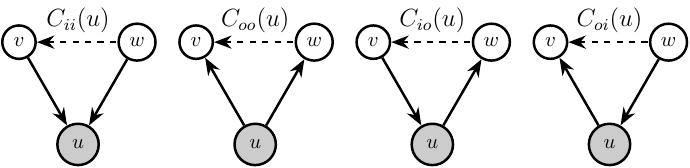} 
\caption{%
Illustration of the local directed clustering coefficients at node $u$, due to 
Fagiolo~\cite{fagiolo2007clustering}. 
The definition is a direct extension of the local clustering coefficient~\cite{watts1998collective},
which measures triadic closure from the center of each wedge.
}
\label{Fig:DefClu}
\end{figure}

Fagiolo proposed a generalization of the clustering coefficient 
to directed networks~\cite{fagiolo2007clustering}. 
Similar to the undirected case, a \emph{directed wedge} is an {\em ordered} pair of edges that share a common node, 
and the common node is called the {\em center} of this wedge.
The wedge is then called \emph{closed} if the there is an edge from the opposite end-point of the second
edge to the opposite end-point of the first edge (this constraint, along with the ordering of the
two edges, covers the symmetries in the problem).
In total, there are four directed clustering coefficients, each defined
by the fraction of certain types of wedges that are closed (\Cref{Fig:DefClu}).
Seshadhri \etal extended Fagiolo's definition by explicitly accounting for
bidirected (reciprocal) edges, with a focus on network-level 
(as opposed to node-level) metrics~\cite{seshadhri2016directed}. 
While we could also explicitly differentiate bidirected links, here we focus
on bidirected links as counting towards wedges and closure in both directions.


Directed clustering coefficients have found applications in analyzing 
fMRI data~\cite{liao2011small},
financial relationships~\cite{minoiu2013network}, 
and social networks~\cite{ahnert2008clustering}.
However, as discussed above, existing
directed clustering coefficients
measure clustering from the center of a wedge, a limited perspective. 
Our head-based directed closure coefficients, which we define formally
in the next section, thus enhance the toolkit for
these diverse applications using triadic closure patterns in directed networks. 

\xhdr{Additional related work}
The first research on directed triadic closure is due to Davis and Leinhardt~\cite{davis1967structure}
who studied the relative frequency of each 3-node directed subgraph pattern and
compared the frequencies with random graph models.
Milo \etal later examined significantly recurring
patterns of connected directed subgraphs as ``network motifs"~\cite{milo2002network},
with a particular emphasis on the role of so-called ``feed-forward loops'' in biology~\cite{mangan2003coherent}.
Similar to the case of directed clustering coefficients~\cite{fagiolo2007clustering}, 
prior research has studied the ratio of closed wedges at the
global (network) level~\cite{onnela2005intensity,brzozowski2011should,seshadhri2016directed}, which is sometimes called ``motif intensity''~\cite{onnela2005intensity}. 
Most generally, we also note that the language of ``directed graphlets''~\cite{sarajlic2016graphlet}, which
can be used to quantify node-level subgraph participation counts within directed networks,
provides an expansive characterization of 128 automorphism orbits 
from which our closure coefficients can be thought of as specific, derived quantities. 
The key differences in our definitions of directed closure coefficients in the next section are that
(i) we measure closure at the node level; and
(ii) they are head-node-based metrics which are more agreeable
with traditional explanations of directed triadic closure.
We will show later that our measures also have considerably different behavior than previous measures.

Directed triadic closure also appears in dynamic network analysis.
Lou \etal proposed a graphical model to
predict the formation of a certain type of directed triadic closure: closing an $oo$-type 
wedge with outbound link~\cite{lou2013learning}. This model was later generalized
to predict the closure of any type of wedge based on node attributes~\cite{huang2014mining}.
Similarly, the notion of a ``closure ratio" has been used to analyze copying 
phenomena in directed networks~\cite{romero2010directed}.
This is also an end-node-based metric that measures a closure of in-in wedges
with an incoming edge. Our definitions of directed closure coefficients are different in that they
(i) are defined on static networks, (ii) measure diverse types of triadic closure,
and (iii) are closely connected to undirected measures of closure and 
the traditional perspective of triadic closure.
Connecting our static measures of directed closure and temporal counterparts is 
an interesting avenue for future research.

%% file: 030def.tex

In this section, we provide our formal definition of directed closure
coefficients, and measure them on some representative real-world networks to
demonstrate how they provide empirical insights. These insights provide
direction and motivation for our theoretical analysis in \cref{sec:theory}.  We
then show in \cref{sec:application} how directed closure coefficients are useful
features in machine learning tasks.

\subsection{Definitions}
With the same motivation as the undirected closure coefficient, we propose to
measure directed triadic closure from the endpoint of a directed wedge.  Recall
that a directed wedge is an ordered pair of edges that share exactly one common
node.  The common node is called the center of the wedge, and here we define
the \emph{head} of this wedge as the other end of the first edge, and
the \emph{tail} as the other end of the second edge. Regardless of the
direction of the edges, we denote a wedge by an ordered node triple $(u, v, w)$, where
$u$ is the head, $v$ is the center, and $w$ is the tail.

Since each edge is directed, there are four types of directed wedges.%
\footnote{Again, for readability purposes, we do not consider reciprocal edges
separately; instead, a reciprocal edges is treated as two separate directed edges.
Our definitions and analyses can easily be extended to
study reciprocal edges, though there would be 9 types of directed wedges and 27 closure coefficients.}
We denote the type of wedge with two
variables, say $x$ and $y$, each taking a value in $\{i, o\}$ to denote incoming
or outgoing.  Specifically, a wedge is of \emph{type} $xy$ (an \emph{$xy$-wedge})
if the first edge is of direction $x$ to the \emph{head},
and the second edge is of direction $y$ to the \emph{center} node.
\Cref{Fig:DefCC} shows the four types of directed wedges.

We say that a wedge is \emph{$i$-closed} if there is an incoming edge from the
tail to the head node, and analogously, it is \emph{$o$-closed} if there is an
outgoing edge from head to the tail node.  For any $u \in V$ and $x,y,z \in \{i,o\}$,
we denote $\wg{xy}{u}$ as the number of wedges of type $xy$
where node $u$ is the head, and $\wgc{xy}{z}{u}$ as the number of $z$-closed wedges of type
$xy$ where node $u$ is the head.

Now we give our formal definition of local directed closure coefficients, which
is also illustrated in \Cref{Fig:DefCC}.

\begin{definition}\label{Def:CC}
The {\bf \emph{local directed closure coefficients}} of node $u$ are eight scalars, 
denoted by $\lcc{xy}{z}{u}$ with $x, y, z \in \{i,o\}$, where
\begin{equation}   \label{Eq:DefLcc}
\lcc{xy}{z}{u} = \frac{\wgc{xy}{z}{u}}{\wg{xy}{u}}.
\end{equation}
If there is no wedge of a certain type with node $u$ being the head, the corresponding 
two closure coefficients are undefined.
\end{definition}

Here we highlight again the fundamental difference between the local directed closure 
coefficients we proposed and the local directed clustering coefficients proposed by Fagiolo~\cite{fagiolo2007clustering}:
the closure coefficients measure triadic closure from the head of wedges, which agrees with natural initiator-driven explanations
on triadic closure, while the clustering coefficients measure from the center of wedges.
We will show that this small definitional difference yields substantial empirical and theoretical disparity.


Analogous to the undirected clustering coefficient, we also define the average
and global directed closure coefficient to measure the overall directed triadic closure
tendency of the network.
\begin{definition}\label{Def:ACC}
The {\bf \emph{average directed closure coefficients}} of a graph are eight scalars, 
denoted by $\acc{xy}{z}$ with $x, y, z \in \{i,o\}$, 
each being the mean of corresponding local directed closure coefficient across the network:
\[ 
\acc{xy}{z} = 
\frac{1}{n} \sum_{u \in  V} \lcc{xy}{z}{u},
\]
We treat local closure coefficients that are undefined as taking the value 0 in this average,
though most nodes in the datasets we analyze have eight well-defined closure coefficients.
\end{definition}


\begin{definition}\label{Def:GCC}
The {\bf \emph{global directed closure coefficients}} of a graph are eight scalars, 
denoted by $\gcc{xy}{z}$ with $x, y, z \in \{i,o\}$, 
each being the fraction of closed directed wedges in the entire network:
\begin{equation}   \label{Eq:DefGcc}
\gcc{xy}{z} = \frac{\wgca{xy}{z}}{\wga{xy}},
\end{equation}
where $\wga{xy} = \sum_{u \in V} \wg{xy}{u}$ and
and $\wgca{xy}{z} = \sum_{u \in V} \wgc{xy}{z}{u}$
are the total number of $xy$-wedges and closed $xy$-wedges.
\end{definition}
The global directed closure coefficients are equivalent to some global 
metrics of directed clustering coefficients~\cite{onnela2005intensity,seshadhri2016directed},
since the difference in measuring from head or center does not surface.

\subsection{Empirical Analysis} \label{sec:empirical}
%

\begin{table*}[t]
  \setlength{\tabcolsep}{3pt}
  \begin{tabular}{l @{\hskip 10pt} c c c c c c c c}
    \toprule
    Network &  $n$ & $m$ & $M_{ii}$ & $M_{io}$ & $M_{oo}$ & $r\%$ & $\Delta_c$ & $\Delta_{ac}$ \\ 
    \midrule
    \lawyer   &  71  &  892  &  227.41  &  166.15  &  208.65 &  0.39  &  880  &  5075\\
    \epinions  &  75.9K  &  509K  &  1179.40  &  526.15  &  721.82  &  0.41  &  740K  &  3.59M \\
    \lj   &  4.85M  &  69.0M  &  2091.52  &  1220.33  &  1504.35  &  0.75  &  244M  &  946M \\
    \midrule
    \college    &  1899  &  20.3K  &  347.80  &  391.99  &  592.42  &  0.64  &  11K  &  40K  \\
    \eu  & 1005  &  25.6K  &  1428.97  &  1509.56  &  1756.77  &  0.72  &  132K  &  433K  \\
    \midrule
    \hepth    &  27.8K  &  353K  &  1746.72  &  269.14  &  416.35  &  0.00  &  572  &  1.49M \\
    \hepph  &  34.5K  &  422K  &  790.63  &  189.62  &  380.70  &  0.00  &  555  &  1.29M   \\
    \midrule
    \everglades  &  69  &  916  &  394.12  &  136.52  &  257.16  &  0.07  &  538  &  4781   \\
    \florida  &  128  &  2106  &  493.08  &  201.92  &  451.62  &  0.03  &  357  &  8688   \\
    \midrule
    \google   &  876K  &  5.11M  &  1572.90  &  69.30  &  77.46  &  0.31  &  3.89M  &  28.2M  \\
    \berkstan   &  685K  &  7.60M  &  62430.80  &  324.71  &  390.55  &  0.25  &  13.8M  &  131M  \\
    \bottomrule
  \end{tabular}
\vspace{5pt}
\caption{Summary statistics of networks: number of nodes $n$; number of edges $m$;
  second-order moments of the degree sequence $M_{ii}$, $M_{io}$, and $M_{oo}$;
  fraction $r$ of edges that are reciprocal (i.e., reciprocity); and
  number of cyclic and acyclic triangles ($\Delta_c$ and $\Delta_{ac}$).}
  \label{tab:stats}
\end{table*}

To obtain intuition and empirical insights before diving into theoretical analysis, 
we evaluate the directed closure coefficients on 11 networks from five different domains:
\begin{enumerate}
\item Three social networks. 
\lawyer~\cite{lazega2001collegial}: a professional advisory network between lawyers in a law firm; 
\epinions~\cite{richardson2003trust}: an online network of who-trusts-whom relationships;
and \lj~\cite{backstrom2006group}: an online social friendship network.
\item Two communication networks. 
\college~\cite{panzarasa2009patterns}: an online messaging network between college students; and
\eu~\cite{yin2017local}: an email network between researchers at a European institute.
\item Two citation networks.
\hepth~and \hepph~\cite{gehrke2003overview}: constructed from arXiv submission in two categories.
\item Two food webs.
\florida~and \everglades~\cite{ulanowicz2005network}: carbon exchange relationships
collected from the Florida Bay and the Everglades Wetland.
\item Two web graphs.
\google~and \berkstan~\cite{leskovec2009community}: hyperlink networks from a Google competition
as well as a crawl of berkeley.edu and stanford.edu domains.
\end{enumerate}
\Cref{tab:stats} lists some basic statistics of the networks.
We emphasize that the reciprocity of these networks vary substantially.
For example, the citation networks and food webs contain mostly unidirectional edges, and
the communication networks many bidirected (reciprocal) edges.

\Cref{fig:GAcc1} shows the global and average directed closure coefficients 
of the \lawyer~dataset.
From the first row, we see that the eight global
closure coefficients can be grouped into four pairs,
\[
\{(\gcc{ii}{i},\gcc{oo}{o}),~(\gcc{ii}{o},\gcc{oo}{i}),
~(\gcc{io}{i},\gcc{io}{o}),~(\gcc{oi}{i},\gcc{oi}{o})\}
\]
with each pair of coefficients taking the same value. This observation is 
expected due to the symmetry in the wedge structure, 
which we study in more detail in \cref{subsec:symmetry}.
In contrast, these groupings do not take the same value in the case
of the average closure coefficients (the second row of \Cref{fig:GAcc1}): 
we observe an {\it a priori} unexpected asymmetry. For example,
$\acc{io}{o} = 0.362 \gg \acc{io}{i} = 0.263$ (in orange, \Cref{fig:GAcc1}).
When an in-out wedge is closed with either an incoming or outgoing edge,
the induced triangle is actually the same: both are feedforward loops \cite{milo2002network}.
It is not obvious why closure with an outgoing edge is so much more likely that with an incoming edge.
We develop some theoretical explanations for this asymmetry in \Cref{subsec:cm_theory}.

\begin{figure}[t]
\centering
\includegraphics[width=\columnwidth]{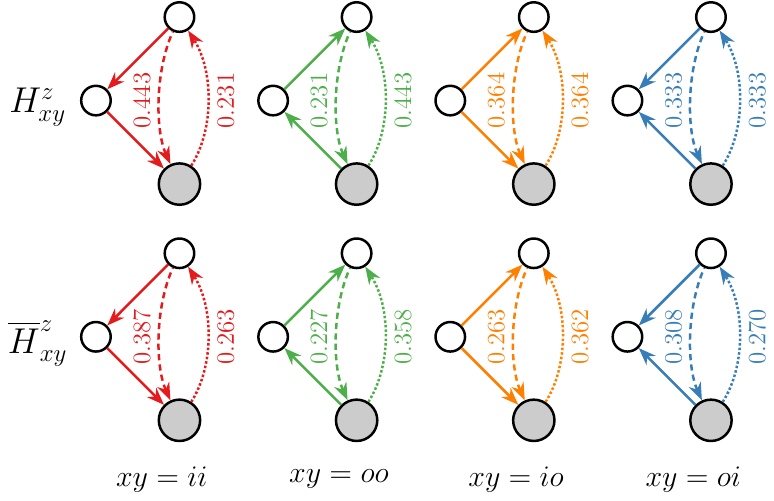}
\caption{%
Global (top) and average (bottom) directed closure coefficients in $\lawyer$,
with head nodes in gray.
The global closure coefficients exhibit symmetry
(\emph{e.g.}, ${\color{black} \gcc{io}{i}} = {\color{black}\gcc{io}{o}}$),
while the average closure coefficients exhibit
counter-intuitive asymmetry between pairs of coefficients, \emph{e.g.},
${\color{black} \acc{io}{i}} = 0.263 \ll {\color{black}\acc{io}{o}} = 0.362$ 
(in orange, second row).
The induced structure is the same in both closure coefficients (a feedforward loop or acyclic triangle).
We explain this phenomenon in \cref{sec:theory}.
}
\label{fig:GAcc1}
\end{figure}

\begin{figure}[t]
\centering
\includegraphics[width=\columnwidth]{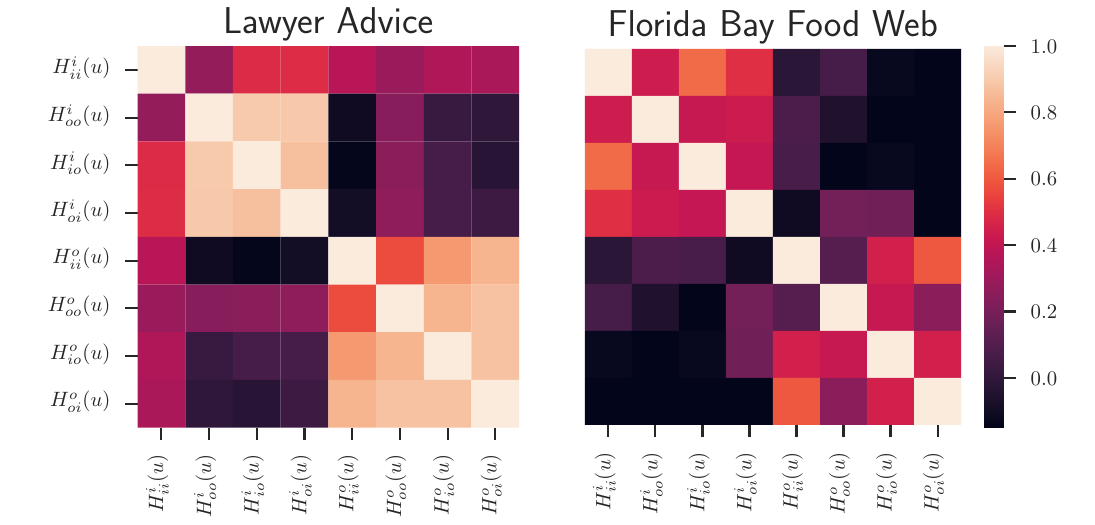} 
\caption{Heatmap of the correlation matrix of the
eight local directed closure coefficients in $\lawyer$ (left) and $\florida$ (right).
There is a clear separation on the eight local closure coefficients: the ones for $i$-closed
and the ones for $o$-closed. Coefficients within each group are highly correlated while
between groups are almost uncorrelated.
}
\label{fig:LccCorr}
\end{figure}

\begin{figure}[h] 
   \centering
\begin{minipage}{0.5\columnwidth}\centering
   {\footnotesize \epinions   }
   \vspace{-10pt}
   \includegraphics[height=0.9in]{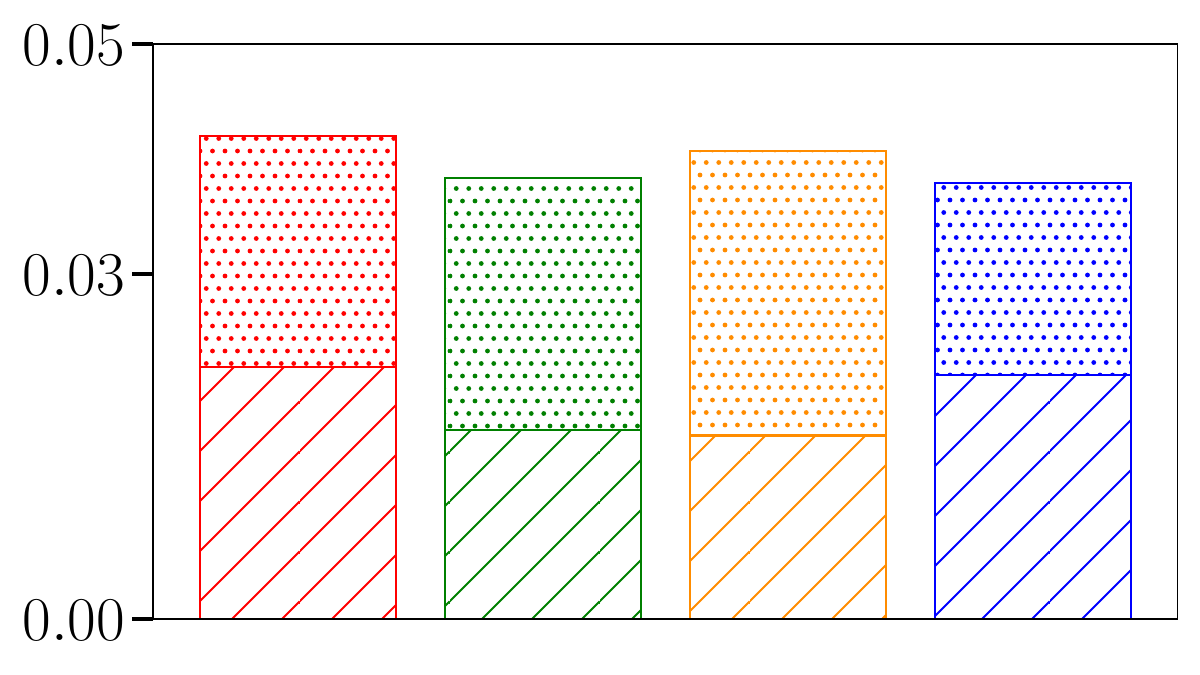} 
\end{minipage}
\begin{minipage}{0.49\columnwidth}\centering
   {\footnotesize \lj   }
   \vspace{-10pt}
   \includegraphics[height=0.9in]{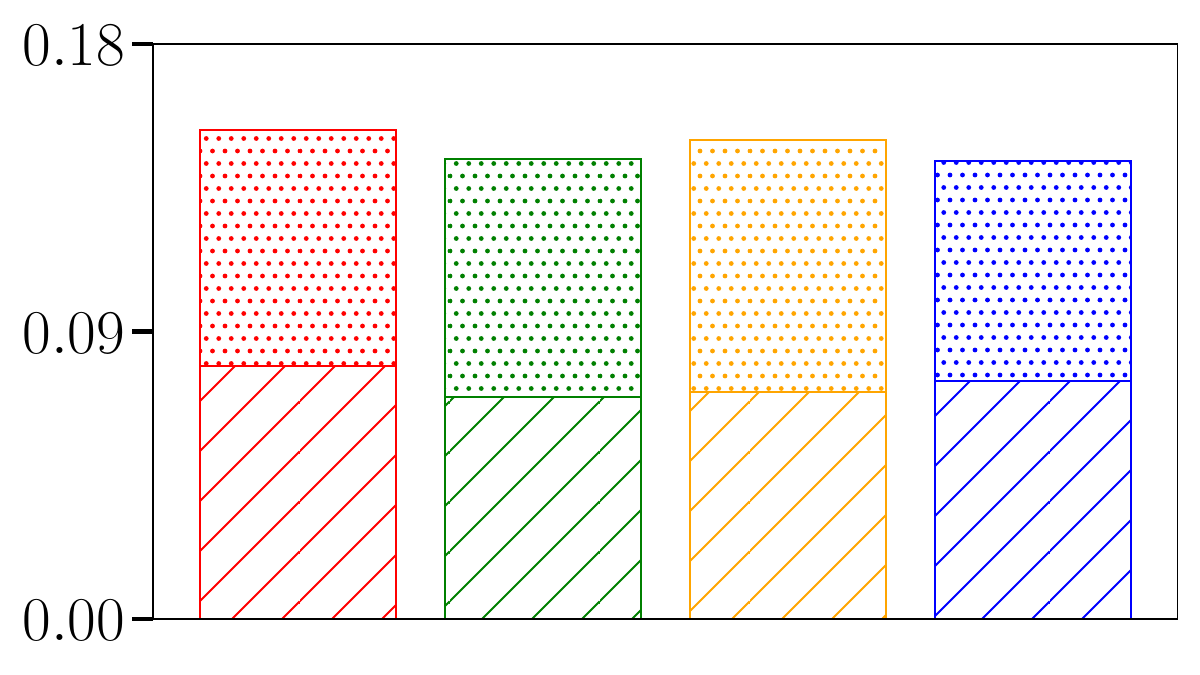} 
\end{minipage}
\\ \vspace{12pt}
\begin{minipage}{0.5\columnwidth}\centering
   {\footnotesize \college   }
   \vspace{-10pt}
   \includegraphics[height=0.9in]{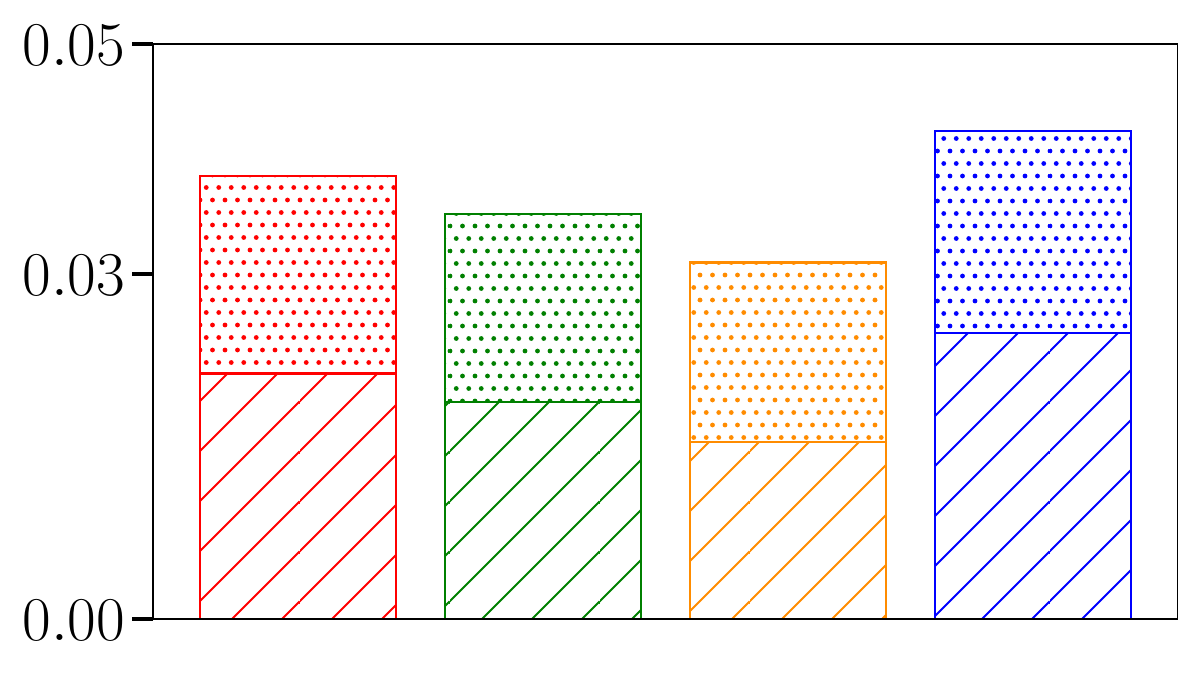} 
\end{minipage}
\begin{minipage}{0.49\columnwidth}\centering
   {\footnotesize \eu   }
   \vspace{-10pt}
   \includegraphics[height=0.9in]{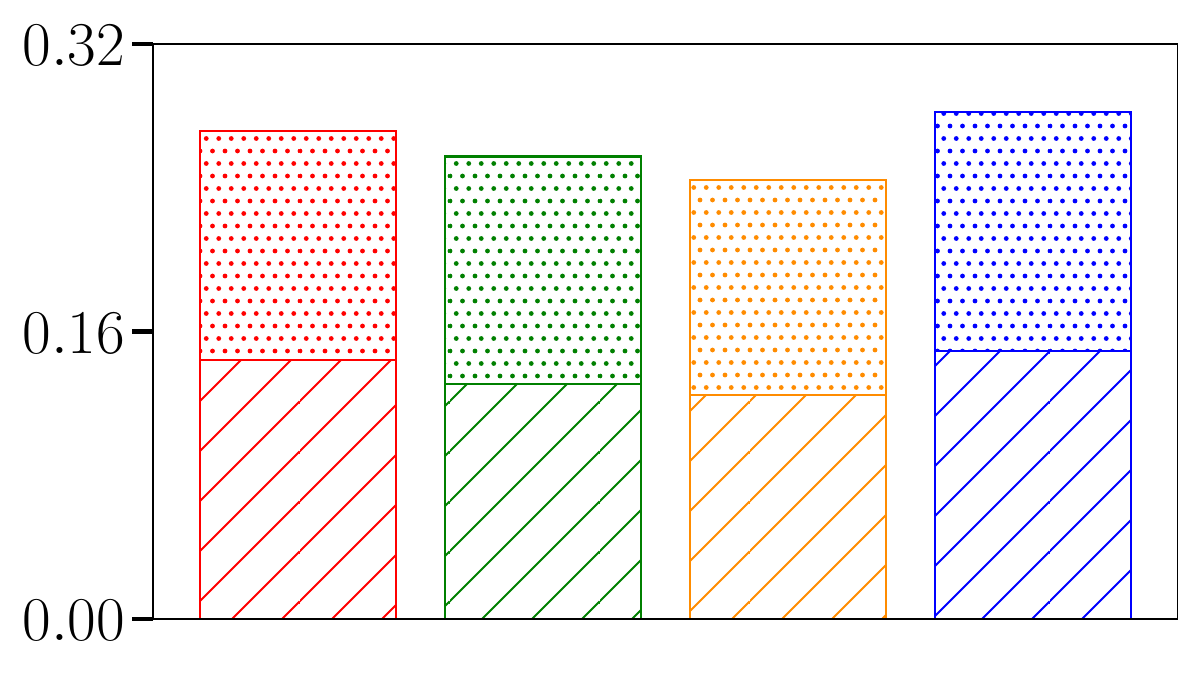} 
\end{minipage}
\\ \vspace{12pt}
\begin{minipage}{0.5\columnwidth}\centering
   {\footnotesize \hepph   }
   \vspace{-10pt}
   \includegraphics[height=0.9in]{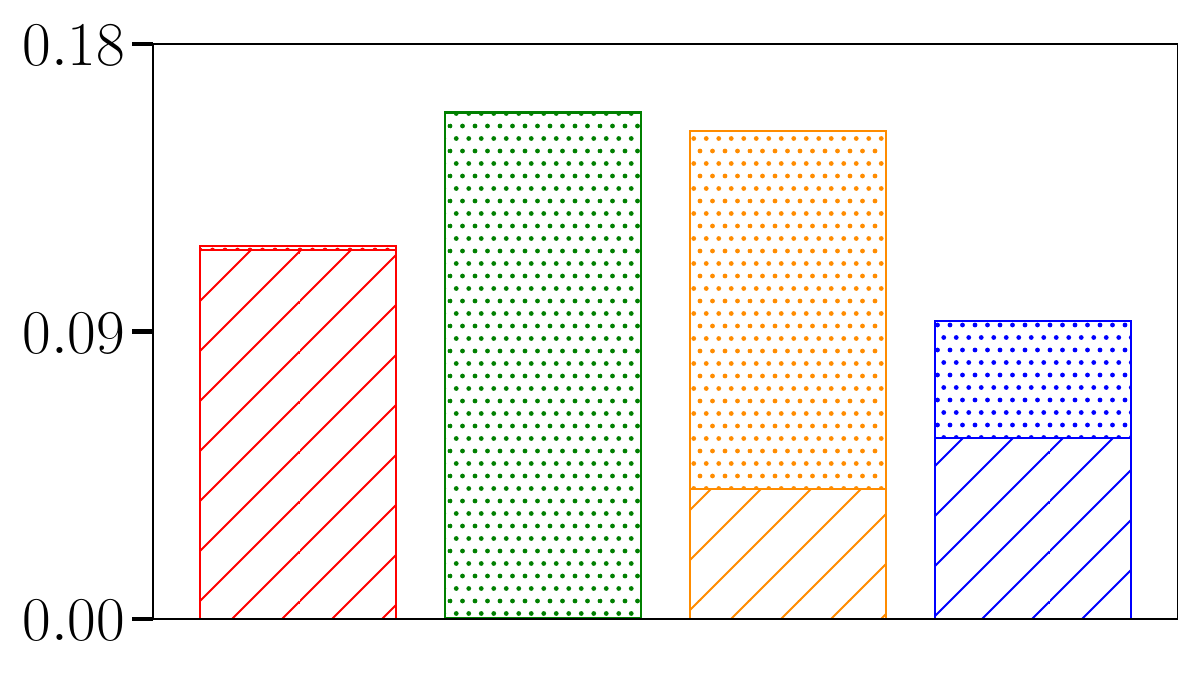} 
\end{minipage}
\begin{minipage}{0.49\columnwidth}\centering
   {\footnotesize \hepth   }
   \vspace{-10pt}
   \includegraphics[height=0.9in]{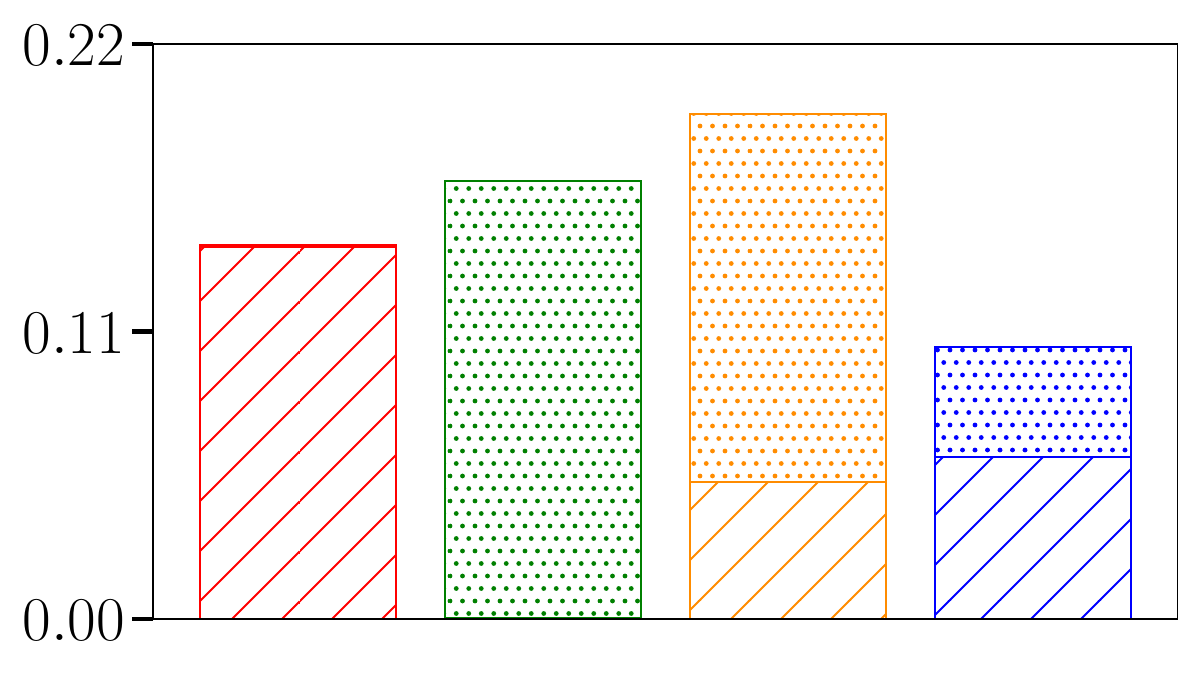} 
\end{minipage}
\\ \vspace{12pt}
\begin{minipage}{0.5\columnwidth}\centering
   {\footnotesize \florida   }
   \vspace{-10pt}
   \includegraphics[height=0.9in]{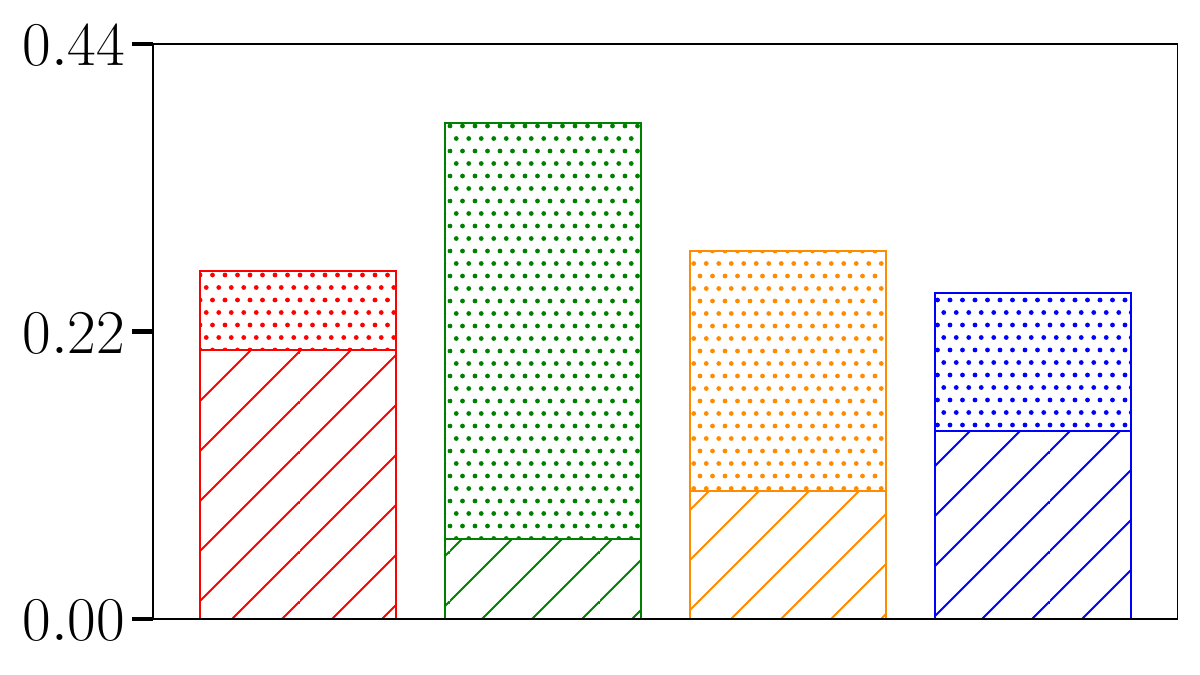} 
\end{minipage}
\begin{minipage}{0.49\columnwidth}\centering
   {\footnotesize \everglades   }
   \vspace{-10pt}
   \includegraphics[height=0.9in]{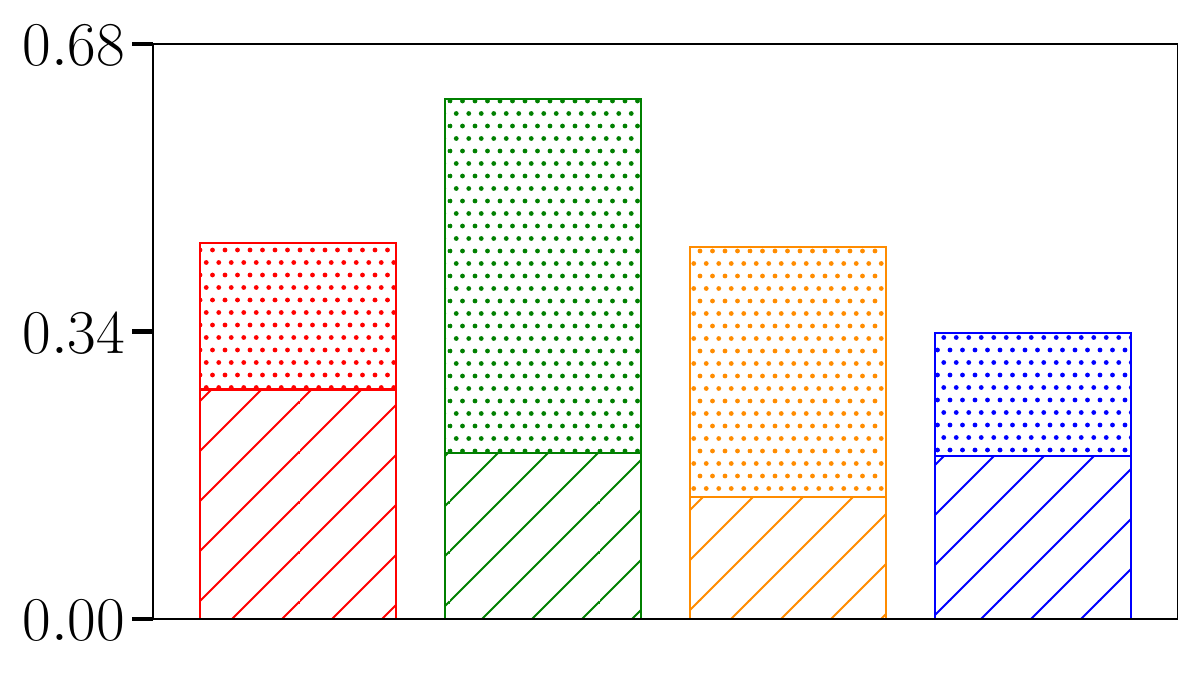} 
\end{minipage}
\\ \vspace{12pt}
\begin{minipage}{0.5\columnwidth}\centering
   {\footnotesize \google  }
   \vspace{-10pt}
   \includegraphics[height=1.2in]{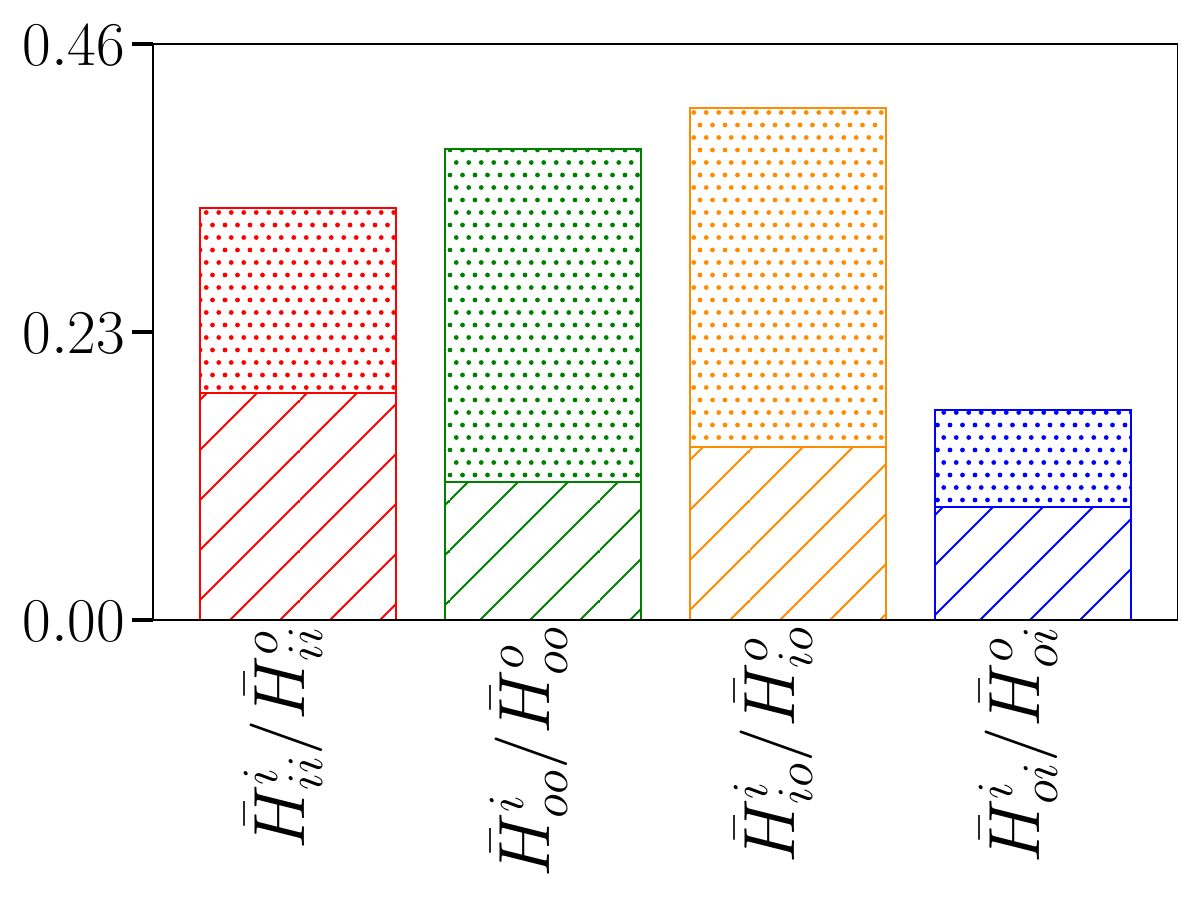} 
\end{minipage}
\begin{minipage}{0.49\columnwidth}\centering
   {\footnotesize \berkstan   }
   \vspace{-10pt}
   \includegraphics[height=1.2in]{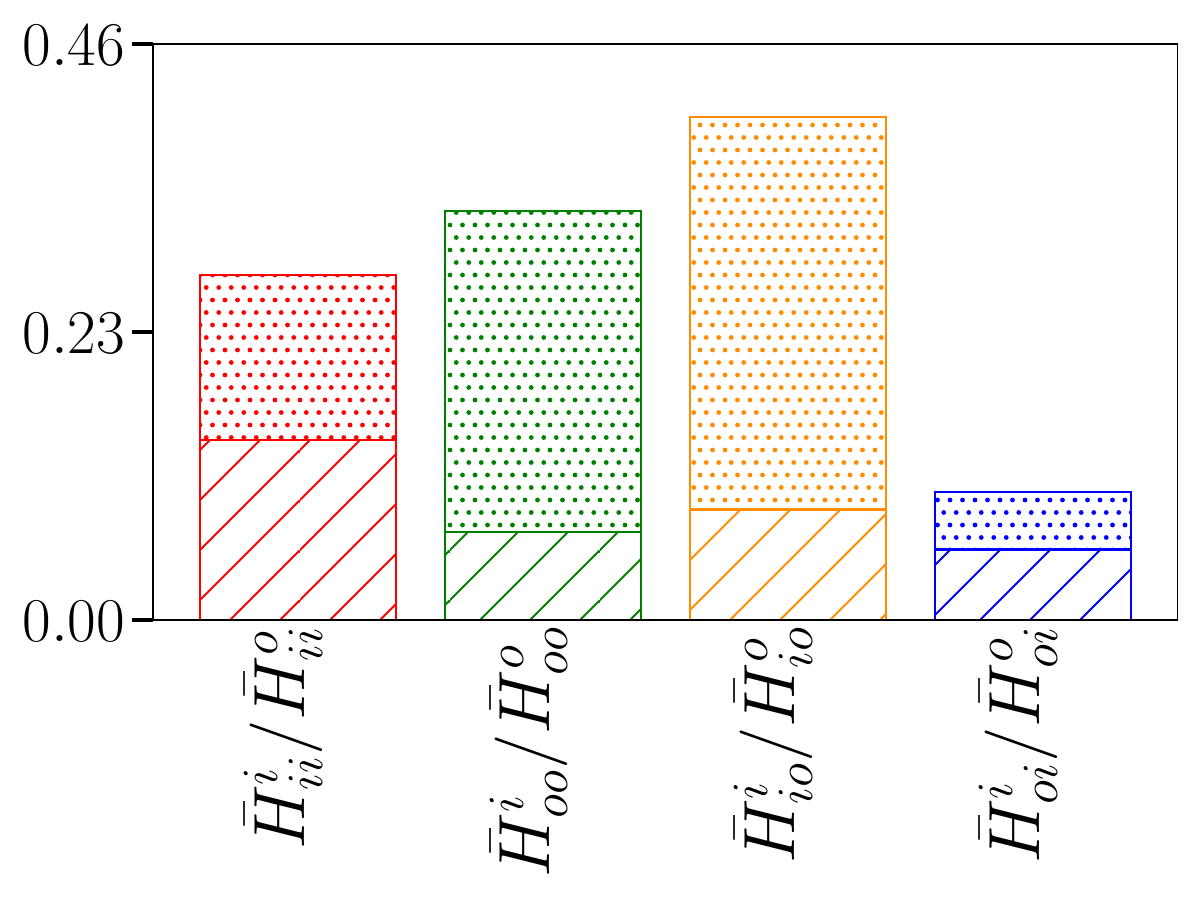} 
\end{minipage}
\\ \vspace{10pt}
\caption{Average directed closure coefficients of networks from five domains. Wedge types 
are colored in the same way as in \Cref{fig:GAcc1}, with incoming closure represented 
by slashed-bars and outgoing closure represented by dotted-bars. Networks from
the same domain (each row) have similar directed closure patterns, while the patterns 
across domains can be quite different.}
\label{Fig:Sesh}
\end{figure}

We also explore the correlations between the eight average directed closure coefficients
in the $\lawyer$ and $\florida$ networks (\Cref{fig:LccCorr}).
Each network has a clear separation amongst the eight local closure coefficients: 
the coefficients in the first four rows and columns (with incoming closure edge),
and the coefficients on the last four rows and columns (with outgoing closure edge).
Within each group, the coefficients are strongly correlated. In
$\lawyer$, the coefficients in different groups are nearly uncorrelated,
whereas in $\florida$, the coefficients in different groups are
negatively correlated. This correlation pattern is representative across the
networks that we have studied with directed closure coefficients, and we explain
this correlation separation as part of the next section.

To study the difference in frequencies of directed triadic closure, we visualize the eight average directed
closure coefficients of 10 networks in \Cref{Fig:Sesh}, where each row contains two networks 
within the same domain. We find that each domain of networks has their own directed triadic
closure patterns. 
In social networks, different wedge types have similar closure frequencies, likely due to the 
abundance of reciprocal edges~\cite{newman2002email}. 
In communication networks, the tall blue bars associated with out-in wedge type means one is more 
likely to connect to people with whom they both send communications; this might be a result of shared interest.
In contrast, citation networks have low closure coefficients for out-in wedge (short blue bars),
meaning that one is not likely to cite or be cited by papers with the same reference: 
this phenomenon might come
from a conflict of interest; moreover, due to near non-existence of cycles,
in-in wedges and out-out wedges are each only closed in one direction.  Similar
patterns appear in the food webs and web graphs, where there is a hierarchical
structure and few cycles. Lastly, we observe similar asymmetry in all
citation, food web, and web graphs, namely that $\acc{io}{o} > \acc{io}{i}$,
from the orange bars showing significantly higher outbound closure rate than
inbound rate.

%% file: 040theory.tex

We now provide theoretical analysis of our directed closure coefficients.  We
first prove the symmetry between the four pairs of global directed closure
coefficients. Motivated by the empirical asymmetry amongst average directed
closure coefficients, we first prove that this asymmetry can be unboundedly
large. Finally, to explain the asymmetry, we study how the in- and out- degree distributions
influence the expected value of each average closure coefficients under a
directed configuration model with a fixed joint degree distribution.

%% file: 041symmetry.tex

\subsection{Symmetry and Asymmetry} 
\label{subsec:symmetry}

Recall that each global directed closure coefficient is the fraction of certain types
of wedges that are closed in the entire network. 
We observed in \Cref{sec:def}
that the eight global directed closure coefficients can be grouped into four pairs,
with each pair of coefficients having the same value. 
The following proposition shows that these values must
be the same in any network.
\begin{proposition} \label{prop:GCC4d}
In any directed network,
$\gcc{ii}{i} = \gcc{oo}{o}$,  $\gcc{ii}{o} = \gcc{oo}{i}$,
$\gcc{io}{i} = \gcc{io}{o}$,  and $\gcc{oi}{i} = \gcc{oi}{o}$.
\end{proposition}
\begin{proof}
Here we only prove $\gcc{ii}{i} = \gcc{oo}{o}$, and the other three identities 
can be shown analogously. By counting wedges 
from the center node, $\wga{ii} = \sum_{u} \indeg{u} \cdot \outdeg{u} = \wga{oo}$.
Next, there is a one-to-one correspondence between a closed in-in wedge
and a closed out-out wedge by flipping the roles of the head and tail nodes.
Thus, $\wgca{ii}{i} = \wgca{oo}{o}$ and $\gcc{ii}{i} = \gcc{oo}{o}$,
according to \Cref{Def:GCC}.
\end{proof}
Proposition~\ref{prop:GCC4d} illustrates the fundamental symmetry among the eight global
directed closure coefficients. The four pairs of global closure coefficients
$$\{(\gcc{ii}{i},\gcc{oo}{o}),~(\gcc{ii}{o},\gcc{oo}{i}),
~(\gcc{io}{i},\gcc{io}{o}),~(\gcc{oi}{i},\gcc{oi}{o})\}$$
correspond to the same structure and triadic closure pattern in the
entire network, so their values have to be the same.

As an alternative global measure of directed triadic closure, we might expect the average
closure coefficients to have a similar symmetric pattern. Specifically, by pairing 
up the average closure coefficients in the same way,
\begin{equation}   \label{Eq:PairedACC}
\{(\acc{ii}{i},\acc{oo}{o}),~(\acc{ii}{o},\acc{oo}{i}),
~(\acc{io}{i},\acc{io}{o}),~(\acc{oi}{i},\acc{oi}{o})\},
\end{equation}
one might initially guess that the two values in a pair would be close.
However, our empirical evaluation on the $\lawyer$ datasets above,
as well as all the citation, food webs, and web graphs,
showed \emph{asymmetry} in these metrics:
for example, $\acc{io}{o} \gg \acc{io}{i}$ in the $\lawyer$ dataset.
Here we study how large such difference can be and find that there is no non-trivial
upper or lower bound on $\acc{io}{i}$ based on $\acc{io}{o}$ and vice versa.
Furthermore, this same flavor of unboundedness is valid for the other three 
pairs of average directed closure coefficients.

\begin{theorem}   \label{Thm:Extreme}
For any $\epsilon > 0$, and any pair of average directed closure coefficients 
from \Cref{Eq:PairedACC},
denoted as $(\acc{a}{}, \acc{b}{})$, there is a finite graph such that
$\acc{a}{} < \epsilon$ and $\acc{b}{} > 1 - \epsilon$,
and another finite graph such that
$\acc{a}{} > 1 - \epsilon$ and $\acc{b}{} < \epsilon$.
\end{theorem}

\begin{figure}[tb]
\begin{center}
\begin{tabular}{c @{\hskip 10pt} c @{\hskip 10pt} c @{\hskip 10pt} c @{\hskip 10pt} c}
\toprule
Class &  \#nodes & $\wg{io}{u}$ & $\wgc{io}{i}{u}$ & $\wgc{io}{o}{u}$ 
\\ \midrule 
\multirow{10}{*}{
        \begin{tikzpicture}
        [
        headnode/.style={circle, draw=black!100, fill=red!50, very thick, minimum size=7mm},
        othernode/.style={circle, double, draw=black!100, very thick, minimum size=1mm},
        baseline=(current bounding box.center)
        ]
        \node[othernode]  (n4)  at (0, 0)        {\scriptsize $C_4$}; 
        \node[othernode]  (n3)  at ($(n4) + (60:1.4)$)   {\scriptsize $C_3$}; 
        \node[othernode]  (n2)  at ($(n3) + (160:1.8)$) {\scriptsize $C_2$}; 
        \node[othernode]  (n1)  at ($(n2) + (45:1.3)$)   {\scriptsize $C_1$}; 
        \draw[-{Stealth[scale=1]}, very thick] (n3)  -- (n4);
        \draw[-{Stealth[scale=1]}, very thick] (n3)  -- (n2);
        \draw[-{Stealth[scale=1]}, very thick] (n2)  -- (n1);
        \draw[-{Stealth[scale=1]}, very thick] (n3)  -- (n1);
        \end{tikzpicture}
}&  \\
& $n_1$  & $n_3n_2 + n_3n_4$ & $n_3 n_2$ & 0 
\\ \\ & $n_2$ & $n_3n_1 + n_3n_4$  & 0 & $n_3 n_1$ 
\\ \\ & $n_3$ & 0 & 0 & 0 
\\ \\ \\ & $n_4$ & $n_3n_1 + n_3n_2$& 0 & 0 
\\ \\[-2mm]
\bottomrule
\end{tabular}
\end{center}
\caption{An example graph used in the proof of \Cref{Thm:Extreme}, showing maximal differences 
between directed closure coefficients $\acc{io}{i}$ and $\acc{io}{o}$.
Each double circle $C_j$ represents a class of nodes and an edge $C_j \to C_k$ means that
$u_j \to u_k$ for all $u_j \in C_j$ and $u_k \in C_k$.
}
\label{Fig:EgExtreme}
\end{figure}%

\begin{proof}
Here we give a constructive proof for the pair $(\acc{io}{i}, \acc{io}{o})$;
the same technique works for the other three pairs.
We use the example graph in \Cref{Fig:EgExtreme}. 
Each double-circle in the figure, denoted by $C_j$ with $j \in \{1,2,3,4\}$, represents a set of nodes, 
and we let $n_j$ denote the number of nodes in each class.
A directed edge from class $C_j$ to $C_k$ means that for any node 
$u_j \in C_j$ and any node $u_k \in C_k$, there is an edge $u_j \rightarrow u_k$.
The number of in-out wedges as well as closed wedges are listed in the last
three columns of the table.
We have that $\lcc{io}{i}{u} = \frac{n_2}{n_2 + n_4}$ for any node $u \in C_1$,
$\lcc{io}{i}{u} = 0$ for $u \in C_2$ or $C_4$, and $\lcc{io}{i}{u}$ undefined for $u \in C_3$. 
Now,
\[ \textstyle
\acc{io}{i} = \frac{n_1 n_2}{(n_2 + n_4)(n_1+n_2+n_3+n_4)},\;\; \acc{io}{o} = \frac{n_1 n_2}{(n_1 + n_4)(n_1+n_2+n_3+n_4)}.
\]
The $n_j$'s can take any integer value. We first fix $n_3 = n_4 = 1$. 
If $n_1 = k^2$ and $n_2 = k$ for any integer $k > 3 / \epsilon$, 
 $\acc{io}{i} > 1 - \epsilon$ and $\acc{io}{o} < \epsilon$.
And if $n_1 = k$ and $n_2 = k^2$ for any integer $k > 3 / \epsilon$, 
$\acc{io}{i} < \epsilon$ and $\acc{io}{o} > 1- \epsilon$.
\end{proof}

In contrast, the directed clustering coefficients due to Fagiolo~\cite{fagiolo2007clustering}
are based on the center of wedges, so the two edges are naturally symmetric and
consequently, the metric is always symmetric. Therefore, there are four
directed clustering coefficients and eight directed closure coefficients.

In the next section, we study how we expect the directed closure coefficients 
to behave in a configuration model, which provides additional insight into
why asymmetries in the directed closure coefficients might be unsurprising.

%% file: 042cm_theory.tex

\subsection{Expectations under Configuration}
\label{subsec:cm_theory}

The previous section showed that pairs of average directed closure coefficients   
can have significantly different values; in fact, our extremal analysis
showed that their ratio can be unbounded in theory.
However, we have not yet provided any intuition for asymmetry in
real-world networks. Here, we provide further theoretical analysis to show  
that the structure of the joint in- and out-degree distribution of a 
network provides one explanation of this asymmetry. When considering
random graphs generated under a directed configuration model with 
a fixed joint degree sequence, 
the coefficients are generally asymmetric even in their expectations.

The configuration model~\cite{molloy1995critical,fosdick2018configuring}
 is a standard tool for analyzing the behavior 
of patterns and measures on networks.
The model is typically studied for undirected graphs with a
specified degree sequence, but the idea cleanly generalizes
to directed graphs with a specified joint degree
sequence~\cite{chen2013directed}. 
It is often hard to understand
the determinants of unintuitive observations on networks. 
What aspect of the specific network under examination leads to 
a given observation? As one specific angle on this question, 
does the observation hold for typical graphs with the observed joint 
degree sequence, and if so, what are the determinants of the behavior?
Analyses using the configuration model can thus be used 
to investigate the expected behavior of a measure, in our
case the directed closure coefficients, under
this random graph distribution.

If a degree sequence satisfies the condition that the maximum degree is upper bounded by
$\sqrt{n}$, then under the configuration model with this degree sequence,  
the probability of forming an edge $u \rightarrow v$ is 
\begin{equation}   \label{Eq:EdgeProb_ConfigM}
\prob[(u, v) \in E \mid S] = \frac{\outdeg{u} \cdot \indeg{v}}{m} \cdot (1 + o(1)),
\end{equation}
where the $o(1)$ term is with respect to large graphs (\ie, $n \rightarrow \infty$)~\cite{newman2003structure}.
As further notation for this section, for an event denoted by $A$, 
we use $\indicator{A}$ as the indicator function
for event $A$, \ie, $\indicator{A} = 1$ when event $A$ happens and $0$ otherwise.
Moreover, we use the symbol ``$\sim$" between two quantities $X \sim Y$ if
$X = Y \cdot (1 + o(1))$.
For any direction variable $x \in \{i, o\}$, we use $\oppo{x}$ to denote the 
opposite direction of $x$.

Before presenting the main theoretical results, we first provide a useful lemma.
The error term $o(1)$ here, as well as those in subsequent theorems in this subsection, vanishes as the size of network grows to infinity,
which is the scenario when the probability of an edge between two nodes in the directed configuration model
is proportional to the product of the nodes' degrees (\Cref{Eq:EdgeProb_ConfigM}).

\begin{lemma}\label{Lem:WedgeEndDegree}
Suppose $G$ is a random directed graph sampled from the directed configuration model 
with joint degree sequence $S$ and let $u$ be any node. 
Let $(u, v, w)$ be a random type-$xy$ wedge with head node $u$.
Then for either direction $z \in \{i, o\}$,
\[
\expect[\dirdeg{z}{w} \mid S, u] = (n / m) \cdot M_{\bar y  z } \cdot (1 + o(1))
.
\]
\end{lemma}
\begin{proof}
%
Conditional on the degree sequence, for any node pair $v^*$ and $w^*$,  $(u, v^*, w^*)$ forms 
an $xy$-wedge with probability
\[
\frac{\dirdeg{x}{u}\dirdeg{\bar x}{v^*}}{m} \cdot \frac{\dirdeg{y}{v^*}\dirdeg{\bar y}{w^*}}{m} \cdot (1 + o(1)) \sim C \cdot \dirdeg{\bar y}{w^*}
\]
where $C$ is a constant independent of $w^*$.
Therefore, for any node $w^*$, it is the other end of a random wedge with probability
$
\prob\left[w = w^* \mid S, u\right]$ $\propto \dirdeg{\bar y}{w^*} \cdot (1 + o(1)),
$
and thus 
\[
\prob\left[w = w^* \mid S, u\right] \sim \frac{\dirdeg{\bar y}{w^*}}{\sum_{w \in V}\dirdeg{\bar y}{w}} = \frac{\dirdeg{\bar y}{w^*}}{m}.
\]
Consequently, we have
\begin{eqnarray*}
\expect[\dirdeg{z}{w} \mid S,u] 
  &=& \sum_{w^* \in V} \dirdeg{z}{w^*} \cdot \prob[w = w^*\mid S, u] 
\\&\sim& \sum_{w^* \in V} \frac{\dirdeg{\bar y}{w^*} \dirdeg{z}{w^*}}{m} = \frac{n M_{\bar y z}}{m}.
\end{eqnarray*}
\end{proof}

Now we present the following theoretical results on the expected value of local
directed closure coefficients under the directed configuration model,
which relates the expected closure coefficient of node $u$ with closing direction $i$
and $o$ to the in- and out-degrees $\indeg{u}$ and $\outdeg{u}$ of $u$.

\medskip\medskip

\begin{theorem}   \label{Thm:CM_LCC}
Let $S$ be a joint degree sequence and $G$ a random directed graph 
sampled from the directed configuration model with $S$.
For any node $u$ and any local directed closure coefficient $\lcc{xy}{z}{u}$,
we have
\begin{equation*}
\expect[\lcc{xy}{z}{u} \mid S] =\frac{n (\dirdeg{z}{u} - \indicator{x=z})}{m^2} \cdot \left(
M_{\bar y \bar z} - \indicator{y=z} \cdot \frac{m}{n} \right) \cdot (1 + o(1)),
\end{equation*}
where $M_{\bar y \bar z}$ is the second-order moment of degree sequence $S$.
\end{theorem}
\begin{proof}
Note that $\lcc{xy}{z}{u}$ can be directly interpreted as the probability that 
a random type-$xy$ wedge $(u,v,w)$ is $z$-closed, where node $u$ is the head of this wedge.
This is the case if there is an edge between $u$ and $w$ of direction $z$ (with respect to
$u$): a $z$-stub from node $u$ is matched to a $\bar z$-stub from node $w$.
Note that the number of $z$-stubs of node $u$ that are not used in wedge $(u, v, w)$ is 
$(\dirdeg{z}{u} - \indicator{x = z})$, where we need to subtract the indicator function
because one $z$-stub is already used in wedge $(u, v, w)$ if $x = z$. 
Similarly, the number of  $\bar z$-stubs of node $w$ 
that are not used in wedge $(u, v, w)$ is $(\dirdeg{\bar z}{w} - \indicator{\bar y = \bar z})$.
According to the setup of the directed configuration model (\Cref{Eq:EdgeProb_ConfigM}), this probability is
\[
{(\dirdeg{z}{u} - \indicator{x = z}) \cdot (\dirdeg{\bar z}{w} - \indicator{\bar y = \bar z})}/{m} \cdot (1 + o(1))
\]
with the given joint degree sequence, and consequently 
\begin{eqnarray*}
\expect[\lcc{xy}{z}{u} \mid S]
&\sim&
\expect[ (\dirdeg{z}{u} - \indicator{x = z}) \cdot (\dirdeg{\bar z}{w} - \indicator{\bar y = \bar z}) / m  \mid S]
\\&=& \displaystyle \frac{\dirdeg{z}{u} - \indicator{x=z}}{m} 
\cdot \left( \expect[\dirdeg{\bar z}{w} \mid S]  - \indicator{\bar y = \bar z}  \right)
\\&\sim& \displaystyle \frac{\dirdeg{z}{u} - \indicator{x=z}}{m} 
\cdot \left( (n / m) \cdot M_{\bar y \bar z }  - \indicator{\bar y = \bar z}  \right)
\\&=& \frac{n (\dirdeg{z}{u} - \indicator{x=z})}{m^2} \cdot \left(
M_{\bar y \bar z} - \indicator{y=z} \cdot \frac{m}{n} \right)
\end{eqnarray*}
where the second step follows from the fact that the only random variable is the 
degree of a random tail node $w$, and the third step is due to Lemma~\ref{Lem:WedgeEndDegree}.
\end{proof}

\Cref{Thm:CM_LCC} shows that the expected value of the local directed closure coefficient
$\lcc{xy}{z}{u}$ increases with $\dirdeg{z}{u}$, the degree in the direction of closure.
One corollary of this result is that under the configuration model the expected values of the 
local closure coefficient with the same closure direction
are all monotonic with the same corresponding degree, and thus they should be correlated
themselves. This result provides one intuition for the block structure of the correlations 
between coefficients found in \Cref{fig:LccCorr}.

We can easily aggregate the results of~\cref{Thm:CM_LCC} to give expected
values of the average directed closure coefficients.

\begin{theorem}   \label{Thm:CM_ACC}
Let $S$ be a joint degree sequence and $G$ be a random directed graph 
generated from the directed configuration model with $S$.
For any average directed closure coefficient $\acc{xy}{z}$,
\begin{equation*}
\expect[\acc{xy}{z} \mid S] = 
\frac{m - n \cdot \indicator{x=z}}{m^2} \cdot \left(
M_{\bar y \bar z} - \indicator{y=z} \cdot \frac{m}{n} \right) \cdot (1 + o(1)).
\end{equation*}
\end{theorem}
\begin{proof}
We have
\begin{eqnarray*}
\expect[\acc{xy}{z} \mid S] 
  & = & \frac{1}{n} \sum_{u} \expect[\lcc{xy}{z}{u} \mid S] 
  \\&\sim& \left(M_{\bar y \bar z} - \indicator{y=z} \cdot \frac{m}{n} \right) 
              \cdot \frac{1}{m^2} \sum_{u} [d_z(u) - \indicator{x = z}]
\\&=& \left(M_{\bar y \bar z} - \indicator{y=z} \cdot \frac{m}{n} \right) 
              \cdot \frac{ m - n \cdot \indicator{x = z}}{m^2},
\end{eqnarray*}
where the second line is due to \Cref{Thm:CM_LCC}.
\end{proof}
\Cref{Thm:CM_ACC} shows that the expected value of any average closure coefficient 
$\acc{xy}{z}$ is mainly determined by $M_{\bar y \bar z}$, a second-order moment 
of the degree sequence. 
In the \lawyer~dataset, we have $M_{io} = 166.15 \ll 227.41 = M_{ii}$, 
meaning that $\expect[\acc{io}{i}] \ll \expect[\acc{io}{o}]$.
This result (partly) explains the asymmetry
observed in \Cref{fig:GAcc1}: the different coefficients are related to 
different moments of the joint degree sequence of the network, at least for graphs sampled 
from a configuration model with different empirical joint degree sequences.

Finally, we can also determine the expected value of global directed closure coefficients
under the configuration model, as given in Theorem~\ref{Thm:CM_GCC}.
Again we first present a useful lemma, which is analogous to 
Lemma~\ref{Lem:WedgeEndDegree}.

\begin{lemma}\label{Lem:WedgeEndDegree_uw}
Suppose $G$ is a random directed graph sampled from the directed configuration model 
with joint degree sequence $S$. 
Let $(u, v, w)$ be a random type-$xy$ wedge, then 
\begin{enumerate}
\item $\expect\left[\dirdeg{\bar z}{w} \mid S \right] = (n / m) \cdot M_{\bar y  \bar z } \cdot (1 + o(1))$;
\item $\expect\left[\dirdeg{z}{u} \mid S \right] = (n / m) \cdot M_{x  z } \cdot (1 + o(1))$;
\item $\expect\left[\dirdeg{z}{u} \dirdeg{\bar z}{w} \mid S \right] =
\expect\left[\dirdeg{z}{u} \mid S \right] \cdot \expect\left[\dirdeg{\bar z}{w} \mid S \right]  \cdot (1 + o(1))$.
\end{enumerate}
\end{lemma}
\begin{proof}
Result 1 is a corollary of Lemma~\ref{Lem:WedgeEndDegree}:
\[
\expect[\dirdeg{\bar z}{w} \mid S]
= \expect\left[ \expect[\dirdeg{\bar z}{w} \mid S, u] \mid S \right]
\sim \frac{n M_{\bar y \bar z}}{m}.
\]
Result 2 is a corollary of result 1: $(u, v, w)$ being
a type-$xy$ wedge is equivalent to $(w, v, u)$ being a type-$\bar y \bar x$ wedge.

Now we show the last result. 
Conditional on the degree sequence, for any node triple $(u^*, v^*, w^*)$, it forms 
an $xy$-wedge with probability
\[
\frac{\dirdeg{x}{u^*}\dirdeg{\bar x}{v^*}}{m} \cdot \frac{\dirdeg{y}{v^*}\dirdeg{\bar y}{w^*}}{m} \cdot (1 + o(1)) 
\sim C \cdot  \dirdeg{x}{u^*} \dirdeg{\bar y}{w^*}
\]
where $C$ is a constant independent of $u^*$ and $w^*$.
Therefore, for any node pair $u^*$ and $w^*$, they are the two ends of a random wedge with probability
$
\prob\left[u = u^*, w = w^* \mid S\right] \propto 
\dirdeg{x}{u^*}  \dirdeg{\bar y}{w^*} \cdot (1 + o(1)),
$
and thus 
\[
\prob\left[u = u^*, w = w^* \mid S\right] \sim \frac{\dirdeg{x}{u^*} \dirdeg{\bar y}{w^*}}{\sum_{u, w \in V}\dirdeg{x}{u}  \dirdeg{\bar y}{w}} = \frac{\dirdeg{x}{u^*} \dirdeg{\bar y}{w^*}}{m^2}.
\]
Consequently, we have
\begin{eqnarray*}
\expect\left[\dirdeg{z}{u} \dirdeg{\bar z}{w} \mid S \right]
  &=&  \sum_{u^*, w^* \in V} \dirdeg{z}{u^*} \dirdeg{\bar z}{w^*} \prob\left[u = u^*, w = w^* \mid S\right]
\\&\sim&   \sum_{u^*, w^* \in V} \dirdeg{z}{u^*} \dirdeg{\bar z}{w^*} \cdot \frac{\dirdeg{x}{u^*} \dirdeg{\bar y}{w^*}}{m^2}
\\&=& \frac{n M_{x z }}{m} \cdot \frac{n M_{\bar y  \bar z }}{m}
\\&\sim& \expect\left[\dirdeg{z}{u} \mid S \right] \cdot \expect\left[\dirdeg{\bar z}{w} \mid S \right],
\end{eqnarray*}
which completes the proof.
\end{proof}

\begin{theorem}   \label{Thm:CM_GCC}
Let $S$ be a joint degree sequence and $G$ be a random directed graph 
generated from the directed configuration model with $S$.
For any global directed closure coefficient $\gcc{xy}{z}$,
\begin{equation*}
\expect[\gcc{xy}{z} \mid S] = 
\left(M_{\bar y \bar z} - \indicator{y=z} \cdot \frac{m}{n} \right) 
\cdot \left(M_{x  z} - \indicator{x=z} \cdot \frac{m}{n} \right) 
\cdot\frac{n^2}{m^3} \cdot (1 + o(1)).
\end{equation*}
\end{theorem}
\begin{proof}
For a random type-$xy$ wedge $(u,v,w)$, we have shown that the probability
of it being $z$-closed is of the order
$(\dirdeg{z}{u} - \indicator{x = z}) \cdot (\dirdeg{\bar z}{w} - \indicator{\bar y = \bar z}) / m$.
Different from the proof of \Cref{Thm:CM_LCC} where node $u$ is fixed, here we do not 
fix node $u$, and meaning that both node $u$ and node $w$ are random. 
\begin{eqnarray*}
\expect\left[\gcc{xy}{z} \mid S \right] 
  &\sim& \frac{1}{m} \cdot {\expect\left[  (\dirdeg{z}{u} - \indicator{x = z}) \cdot (\dirdeg{\bar z}{w} - \indicator{\bar y = \bar z})  \mid S \right] }
\\&\sim& \frac{1}{m} \cdot {\left( \expect[\dirdeg{z}{u} \mid S]  - \indicator{x=z} \right) \left( \expect[\dirdeg{\bar z}{w} \mid S] - \indicator{\bar y = \bar z} \right)  }
\\ &\sim& \frac{1}{m} \cdot \left( \frac{n}{m} \cdot M_{xz} - \indicator{x=z} \right)  \cdot \left( \frac{n}{m} \cdot M_{y \bar z} - \indicator{\bar y = \bar z} \right),
\end{eqnarray*}
where the second line is due to the last result in Lemma~\ref{Lem:WedgeEndDegree_uw},
and the last line is due to the first two results in Lemma~\ref{Lem:WedgeEndDegree_uw}.
\end{proof}

As a byproduct of our analysis,
the proof of \Cref{Thm:CM_GCC} also shows
that, under the directed configuration model, the probability that a wedge is
closed is independent of the center node, and thus equal to the network-level
average. This observation gives us the expected value of Fagiolo's directed local clustering
coefficients~\cite{fagiolo2007clustering} under this random graph model as well.

\begin{proposition}   \label{prop:CM_clustering}
Let $S$ be a joint degree sequence and $G$ be a random directed graph 
generated from the directed configuration model with $S$.
For local directed clustering coefficient $\lccc{xy}{u}$,
\begin{equation*}
\expect[\lccc{xy}{u} \mid S] = \expect[\gcc{\bar x y}{i} \mid S].
\end{equation*}
\end{proposition}
\begin{proof}
Consider a random wedge with $u$ being the center node, $(v, u, w)$, which is
an $\bar x y$-wedge to node $v$.
From the definition of $\lccc{xy}{u}$, this wedge is closed if it is $i$-closed to node $v$.
Since the probability of this wedge being $i$-closed is independent of node $u$, 
it is the same as if we randomly choose a wedge without constraining
node $u$ as the center, and thus
$
\expect[\lccc{xy}{u} \mid S] = \expect[\gcc{\bar x y}{i} \mid S].
$
\end{proof}

\begin{figure*}[t] 
   \centering
   {\large \lawyer} \\
   \includegraphics[width=6in]{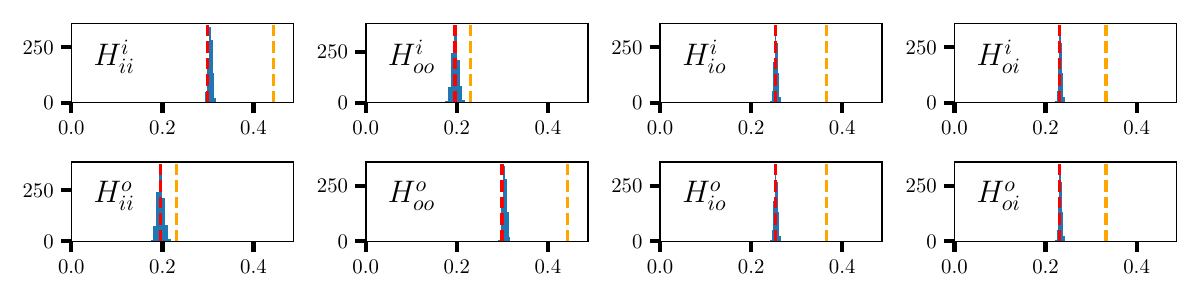} \\
   \includegraphics[width=6in]{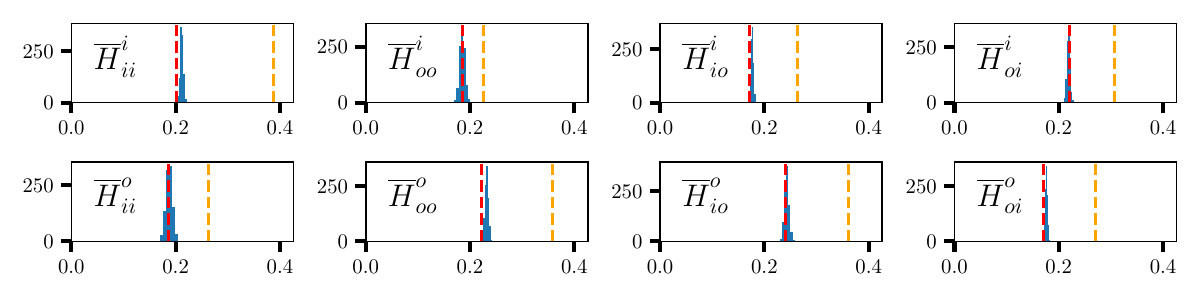} \\
   {\large \hepph}\\
   \includegraphics[width=6in]{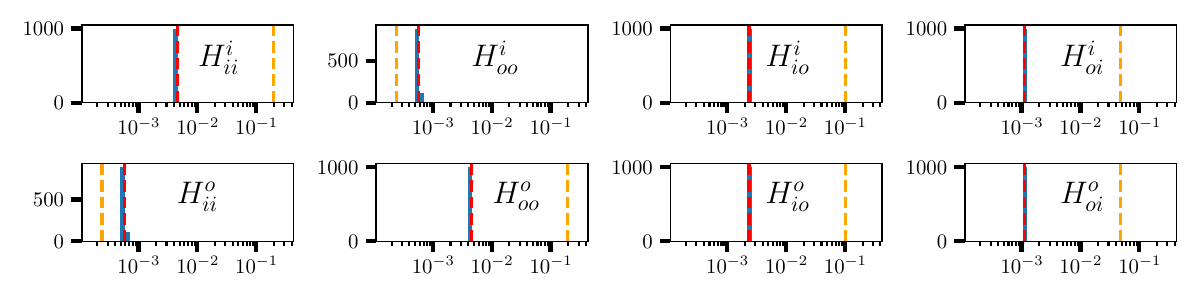} \\
   \includegraphics[width=6in]{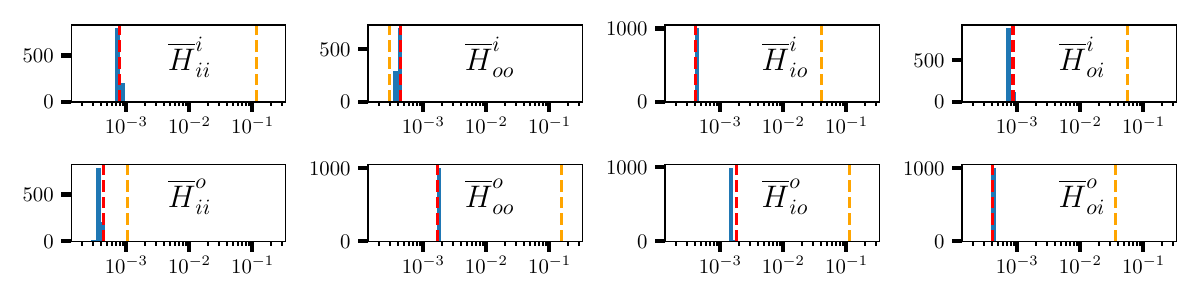} 
   \caption{Histogram of each global closure coefficient (first two rows) and average closure 
   coefficient (last two rows) in 1,000 directed configuration model random graphs with 
   the joint degree sequence of the \lawyer~and \hepph~network. The $x$-axis is the value 
   of various directed closure coefficients and the $y$-axis is the frequency. Besides
   the histogram, we also plot the expected value of closure coefficients from 
   \Cref{Thm:CM_ACC,Thm:CM_GCC}~(red) as well as the actual value in the
   original network (orange).
   }
   \label{Fig:CM_GACC}
\end{figure*}

Next, we study the accuracy of the theoretical expected values
of the average and global closure coefficients under the directed configuration model. 
The directed configuration model can be sampled by using
double-edge swaps~\cite{rao1996markov}. To sample graphs from the model,
we begin with an empirical graph (the graph of interest) 
with joint degree sequence $S$. We then 
select a pair of random directed edges to swap, which changes the graph
slightly but notably preserves the degree sequence. Taking care to avoid
self-loops and multi-edges~\cite{fosdick2018configuring}, the double-edge swap
can be interpreted as a random walk in the space of simple graphs 
with the same degree sequence, and the stationary distribution of this random walk 
is the uniform distribution over the network space.
The swap is then repeated many times to generate graphs that are sampled from
the stationary distribution. The mixing time of these random walks are generally believed 
to be well-behaved, but few rigorous results are known~\cite{greenhill2014switch}.


We generate 1,000 random graphs with the same joint degree sequence as the 
\lawyer~and \hepph~network; to generate each graph,
we repeat the edge-swapping procedure 10,000 times.
\Cref{Fig:CM_GACC} shows histograms of the distribution of each
average and global closure coefficient under this configuration model. 
We see that our approximate formulas from \Cref{Thm:CM_ACC,Thm:CM_GCC} 
are very accurate even when the network is only moderate in size ($n=71$
for \lawyer). 
The theoretical formulas are only guaranteed to be accurate on large sparse
networks, and we do observe a small difference between the expected and
simulated means (\emph{e.g.}, $\acc{ii}{i}$).

The simulation shows that the average and global closure coefficients
have low variance under this configuration model, 
and the values in the original network deviate significantly from these distributions.
More specifically, the values in the citation network are mostly larger than
the distribution by orders of magnitude, and the exceptions are $\gcc{oo}{i}$ and $\gcc{ii}{o}$,
as well as their average closure coefficient counterparts, where the real-world values are low
due to the natural lack of cycles in the nature of citation networks.
This provides evidence that the directed closure coefficients of real-world
networks capture interesting empirical structure
beyond what one would expect from a graph drawn uniformly at random from 
the space of graphs with the same joint degree sequence.

%% file: 050application.tex

Now that we have a theoretical understanding of our directed closure
coefficients, we turn to applications. Directed closure coefficients are a new
measurement for directed triadic closure and thus can serve as a feature for
network analysis and inference. In this section, we present two illustrative
examples to exhibit the strong predictive potential in directed closure
coefficients. Specifically, we present two case studies of node-type
classification tasks, where we show the utility of local directed closure
coefficients in predicting the node type in the $\lawyer$ and $\florida$ datasets analyzed above.
By using an interpretable regularized model, we are able to identify the
salient directed closure coefficients that are useful for prediction.
This analysis reveals new social status patterns in the social network 
and also automatically identifies previously-studied triadic patterns in food webs as good predictors.

\subsection{Case Study I: Identifying Lawyer Status in an Advisory Network}
The $\lawyer$ dataset collected by Lazega~\cite{lazega2001collegial} is a social network of lawyers
at a corporate firm. There is a node for each of the 71
lawyers, and each is labeled with a status level---\emph{partner} or \emph{associate}.
Of the 71 lawyers in the dataset, 36 are partners and 35 are associates.
The edges come from survey responses on who individuals go to
for professional advice: there is an edge from $i$ to $j$ if person $i$
went to person $j$ for professional advice.
Of the edges, there are
395 between two partners;
196 between two associates;
59 from partner to associate; and
242 from associate to partner.

In this case study, our goal is to predict the status of the lawyers (associate or partner) 
with predictors extracted from the advice network.
We consider the following six sets of network attributes as predictors:
\begin{enumerate}   [(1)]
\item \textbf{degree}: the in- and out-degree, and the number of reciprocal edges at each node;
\item \textbf{degree + 1-hop}: the union of the degree predictors and
four neighbor-degrees: the average in- and out-degree of all in- and out-neighbors of each node;
\item \textbf{closure}: the eight local directed closure coefficients defined in this paper;
\item \textbf{closure + degree}: the union of the closure coefficients and the degree predictors;
\item \textbf{clustering}: the four local directed clustering coefficients as defined by 
Fagiolo~\cite{fagiolo2007clustering}; and
\item \textbf{clustering + degree}: the union of the local directed clustering coefficients and the degree predictors.
\end{enumerate}

\begin{table}[t]
  \caption{Validation set accuracy and AUC in classifying node types
  in the $\lawyer$ dataset (partner  vs.\ associate).
  Our proposed local directed closure coefficients are the best set of predictors,
  illustrating the utility of directed closure coefficients in node-level prediction tasks. 
  In contrast, the local directed clustering coefficients~\cite{fagiolo2007clustering} are not as effective.}
  \setlength{\tabcolsep}{3pt}
  \begin{tabular}{r c c c c c c}
    \toprule
             & \small{degree} & \small{degree} & \small{closure}         & \small{closure}   & \small{clustering} & \small{clustering} \\
             &       & \small{+ 1-hop} &                 & \small{+ degree} &            & \small{+ degree}   \\
    \midrule
    \small{accuracy} & 0.7884 & 0.8270 & \textbf{0.8743} & 0.8585   & 0.6255     & 0.7884     \\
    \small{AUC}      & 0.8763 &  0.8978 & \textbf{0.9235} & 0.9183   & 0.6362     & 0.8765     \\
  \bottomrule    
  \end{tabular}
  \label{tab:accuracy_lawyer}
\end{table}

\begin{figure}[tb]
   \centering
   \vspace{5pt}
   \includegraphics[width=\columnwidth]{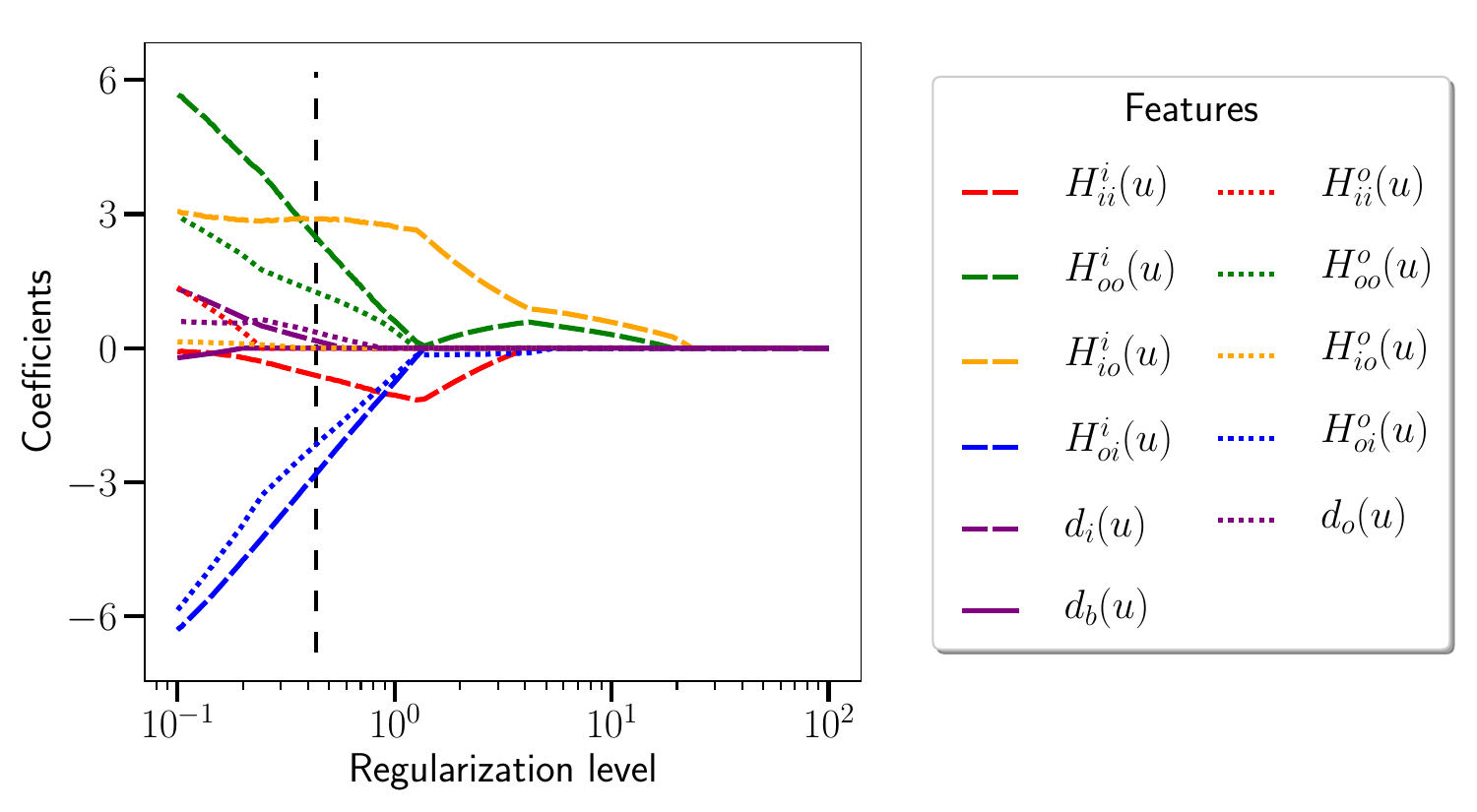} 
   \vspace{-3pt}
   \caption{Regularization path of the $\ell_1$-regularized logistic regression model 
   with predictor set \textbf{closure + degree}  for the model of the $\lawyer$ dataset.
   The $x$-axis is the regularization level, 
   and the $y$-axis is the regression coefficient for each predictor. The vertical black  
   dashed line represents the optimal regularization level obtained from cross-validation.
   The degree attributes are only selected at very low regularization levels, 
   and various local directed closure coefficients dominate the prediction model.}
   \label{fig:RegPath_lawyer}
\end{figure}

For each predictor set, we use 100 random instances of 3-fold cross-validation
to select an $\ell_1$-regularized logistic regression model for predicting whether
or not a node is a partner (\ie, the positive label is for partner).
\Cref{tab:accuracy_lawyer} reports validation set accuracy and AUC.
(Even though different predictor sets have different dimensions,
evaluating the performance in this way makes them comparable.)
The predictors that include our local directed closure coefficients substantially
outperform the other predictor sets. The predictor
set that includes both degrees and closure coefficients slightly underperforms
the one with only closure coefficients, indicating slight overfitting in the
training data, which implies that the degree attributes provides redundant and
noisy information in addition to the closure coefficient attributes in this
prediction task.
In contrast, adding the \textit{1-hop degrees} does not give as much improvement
as they do not consider the triadic closure factors.

To understand how the directed local closure coefficients improve prediction performance,  
we analyze the regularization path of our model, a standard method in sparse regression 
to visualize the predictors at each regularization level~\cite{friedman2010regularization}. 
\Cref{fig:RegPath_lawyer} shows the regularization path for the predictor
set that includes both the local directed closure coefficients and the degree predictors.

We highlight a few important observations.
First, as regularization is decreased, directed local closure coefficients are 
selected before the degree predictors, indicating that the closure coefficients are more relevant in prediction than degrees.
Second, the two predictors with largest positive coefficients at the optimal level of regularization
are $\lcc{io}{i}{u}$ and $\lcc{oo}{i}{u}$, meaning that
lawyers with partner status are more likely to advise people who also advise others. 
In contrast, the in-degree $\indeg{u}$ predictor is not one of the first selected, which implies that
it is not \emph{how many one advises} but rather \emph{who one advises}
that is correlated with partner status.
Finally, the two predictors with the largest negative coefficients at the
optimal regularization are $\lcc{oi}{i}{u}$ and $\lcc{oi}{o}{u}$, meaning that
partner-status lawyers are less likely to interact with other lawyers with whom
they share an advisor.

\begin{table}[t]
  \caption{Validation set accuracy and AUC in classifying node types
  in the $\florida$ dataset (fish vs.\ non-fish).
  Our proposed local directed closure coefficients are again the best set of predictors
  (see also \cref{tab:accuracy_lawyer}),
  illustrating the utility of directed closure coefficients in node-level prediction tasks
  outside of social network analysis.}
  \setlength{\tabcolsep}{3pt}
  \begin{tabular}{r c c c c c c}
    \toprule
             & \small{degree} & \small{degree} & \small{closure}         & \small{closure}   & \small{clustering} & \small{clustering} \\
             &       & \small{+ 1-hop} &                 & \small{+ degree} &            & \small{+ degree}   \\
    \midrule
    \small{accuracy} & 0.6250 & 0.8466  & \textbf{0.8735} & 0.8700   & 0.6875     & 0.7366     \\
    \small{AUC}      & 0.6772 & 0.9127  & \textbf{0.9538} & 0.9529   & 0.7472     & 0.7834     \\
  \bottomrule    
  \end{tabular}
  \label{tab:accuracy_florida}
\end{table}

\begin{figure}[t]
   \centering
   \vspace{5pt}
   \includegraphics[width=\columnwidth]{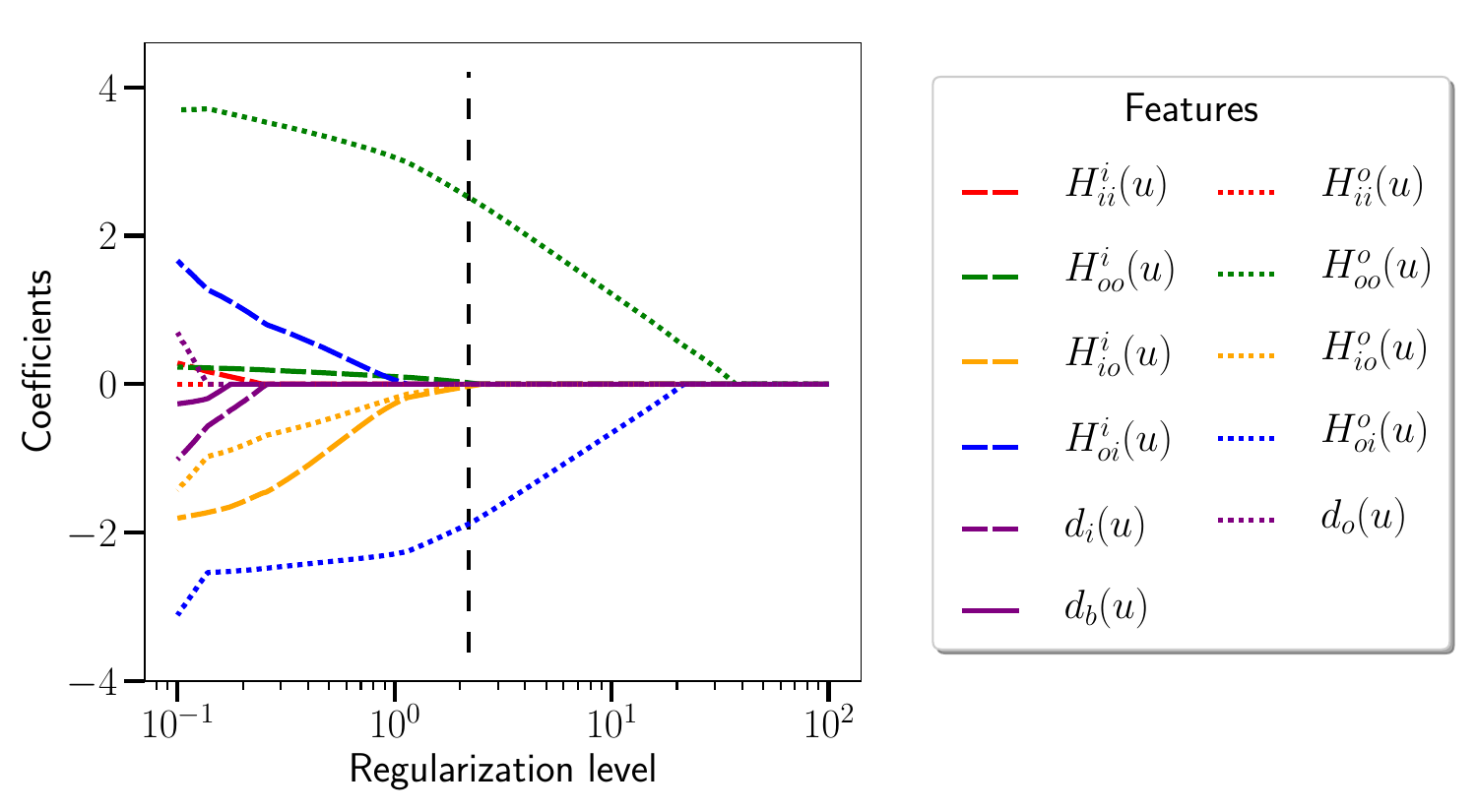} 
   \vspace{-3pt}
   \caption{Regularization path of the $\ell_1$-regularized logistic regression model 
   with predictor set \textbf{closure + degree} for the model of the $\florida$ dataset.
   The $x$-axis is the regularization level, 
   and the $y$-axis is the regression coefficient for each predictor. The vertical black  
   dashed line represents the optimal regularization level obtained from cross-validation.
   The degree attributes are only selected at very low regularization levels, 
   and various local directed closure coefficients dominate the prediction model.}
   \label{fig:RegPath_florida}
\end{figure}

\subsection{Case Study II: Identifying Fish in a Food Web}
We now perform a similar network prediction task.
Here, the data comes from an entirely different domain (ecology),
but we still find that our local directed closure coefficients are effective
predictors for identifying node type.

More specifically, we study a food web collected from the Florida Bay~\cite{ulanowicz2005network}.
In this dataset, nodes correspond to ecological compartments (roughly, species)
and edges represent directed carbon exchange (roughly, who-eats-whom).
There is an edge from $i$ to $j$ if energy flows from compartment $i$ to compartment $j$.
There are 128 total compartments, of which 48 correspond to fish. Our prediction
task in this case study is to identify which nodes are fish using basic node-level features.
The dataset contains 2,106 edges, of which
268 are between fish;
699 are between non-fish;
648 are from a fish to a non-fish; and
491 are from a non-fish to a fish.

We used the same model selection procedure as in the first case study on the $\lawyer$
dataset described above.
\Cref{tab:accuracy_florida} lists the accuracy and AUC of the $\ell_1$-regularized logistic
regression model. We again find that our proposed directed closure coefficients form the best
set of predictors for this task. We also find minimal difference in prediction accuracy when including
degree features, indicating that the degree features provide little predictive information beyond
the directed closure coefficients.

In fact, the regularization path shows that the two closure coefficients $\lcc{oo}{o}{u}$ and $\lcc{oi}{o}{u}$ 
are the most important predictors for identifying fish (\Cref{fig:RegPath_florida}), the former being
positively correlated with the fish type and the latter positively correlated with the non-fish type.
The type of closure associated with the coefficient $\lcc{oo}{o}{u}$ has previously been identified
as important for the network dynamics of overfishing~\cite{bascompte2005interaction}, so it is reasonable
that this predictor is important.

%% file: 070discuss.tex

Triadic closure and local clustering are fundamental properties of complex networks.
Although these concepts have a storied history, only recently have there been
local closure measurements (for undirected graphs) that accurately reflect the ``friend of friend''
mechanism pervasive in discussions of closure. In this paper, we have extended the
subtle definitional difference of initiator-based vs.\ center-based clustering
to directed networks, where clustering in general has received relatively little attention. 
We observed a seemingly
counter-intuitive result that the same induced triadic structure can produce two
different average directed closure coefficients; however, this asymmetry is
understandable through our analysis of closure coefficients within a configuration
model, which points to the role of moments of the in- and out-degree distributions.
Additional analysis showed that this asymmetry can be arbitrarily large.  

One of the benefits of these new local network measurements is that
they can be used as predictors for statistical inference on networks.
Two case studies showed that our directed closure
coefficients are good predictors at identifying node types in two starkly 
different domains--social networks and ecology--with simple models
using these features achieving over 92\% mean 
AUC in both cases. Furthermore, directed closure coefficients are much better 
predictors than generalizations of clustering coefficients to directed graphs 
for these tasks. We anticipate that closure coefficients will become a useful tool 
for understanding the basic local structure of directed complex networks.